\documentclass[letter,11pt,final]{article}
\usepackage{typearea}
\typearea{14}

\usepackage[colorlinks=true,citecolor=blue,urlcolor=blue,bookmarks=false]{hyperref}
\usepackage[numbers,sort]{natbib}
\usepackage{amsfonts} % e.g. mathbb
\usepackage{amsmath,amsthm,amssymb}
\usepackage{thmtools}
\usepackage{thm-restate}
\usepackage{cleveref}
\usepackage{xcolor}
\usepackage{xspace}
\usepackage{bm}
\usepackage{todonotes}
\usepackage{enumitem}
\usepackage{float} % forces figures exactly where you put them

\usepackage{algorithm} % for writing algorithms
\usepackage[noend]{algpseudocode}
\usepackage[labelfont=bf]{caption}

\theoremstyle{plain}
\newtheorem{theorem}{Theorem}[section]

\newtheorem{lemma}[theorem]{Lemma}

\newtheorem{definition}[theorem]{Definition}

\newcommand{\eps}{\varepsilon}
\newcommand{\E}{\mathbb{E}}

\newcommand{\calP}{\mathcal{P}}
\newcommand{\calS}{\mathcal{S}}
\newcommand{\poly}{\mathrm{poly}}

\newcommand{\bigO}[1]{\mathcal{O}\!\left(#1\right)}
\newcommand{\tildeO}{\widetilde{\mathcal{O}}}

\newcommand{\dil}{D}
\newcommand{\cng}{C}
\newcommand{\zeroonedemand}{$\{0, 1\}$-demand\xspace}
\newcommand{\zeroonedemands}{$\{0, 1\}$-demands\xspace}
\newcommand{\Z}{\mathbb{Z}}
\newcommand{\OPT}{\mathrm{OPT}}
\newcommand{\pth}{\mathrm{path}}
\newcommand{\seq}{\mathrm{seq}}
\newcommand{\seqset}{M^{*}}
\newcommand{\state}{\mathrm{state}}

\newcommand{\ind}{\mathrm{ind}}
\newcommand{\virt}{\mathrm{virt}}
\newcommand{\len}{\mathrm{len}}
\newcommand{\pos}{\mathrm{pos}}
\newcommand{\que}{\mathrm{que}}
\newcommand{\seqn}{\mathrm{seq}}

\newcommand{\compl}{\mathrm{compl}}
\newcommand{\greedycompl}{\mathrm{compl}_{\mathrm{greedy}}}

\newcommand{\src}{\mathrm{source}}
\newcommand{\tar}{\mathrm{sink}}
\newcommand{\rev}{\mathrm{rev}}
\newcommand{\ps}{X}
\newcommand{\pa}{x}
\newcommand{\js}{J}
\newcommand{\jo}{j}

\makeatletter
\newcommand{\algmargin}{\the\ALG@thistlm}
\makeatother
\algnewcommand{\parState}[1]{\State%
    \parbox[t]{\dimexpr\linewidth-\algmargin}{\strut\hangindent=\algorithmicindent \hangafter=1 #1\strut}}

\renewcommand{\todo}[1]{}
\newcommand*{\alert}[1]{}
\newcommand{\goran}[1]{}
\newcommand{\antti}[1]{}
\newcommand{\shyr}[1]{}
\newcommand{\cliff}[1]{}
\title{Polylog-Competitive Deterministic Local Routing and Scheduling}

\author{
    Bernhard Haeupler \thanks{Supported in part by the European Research Council (ERC) under the European Union's Horizon 2020 Research and Innovation Programme (grant agreement No.~949272).}\\
    \small ETH Z\"urich \& CMU\\
    \small \texttt{bernhard.haeupler@inf.ethz.ch}\\
    \and
    Shyamal Patel \thanks{Supported in part by NSF grants IIS-1838154, CCF-2106429, CCF-2107187, CCF-2218677, ONR grant ONR-13533312, and an NSF Graduate Student Fellowship.}\\  
    \small Columbia University\\
    \small \texttt{shyamalpatelb@gmail.com}\\
    \and
    Antti Roeyskoe \footnotemark[1]\\
    \small ETH Z\"urich\\
    \small \texttt{antti.roeyskoe@inf.ethz.ch}
    \and
    Cliff Stein \thanks{Supported in part by NSF grant CCF-2218677, ONR grant ONR-13533312 and by the Wai T. Chang Chair in Industrial Engineering and Operations Research}\\
    \small Columbia University\\
    \small \texttt{cliff@ieor.columbia.edu}\\
    \and
    Goran Zuzic\\
    \small Google Research\\
    \small \texttt{goranzuzic@google.com}\\
}

\date{\today}
\date{}

\begin{document}
\maketitle
\thispagestyle{empty}

%\alert{add a table of contents}

% \begin{abstract}
%    We prove that for any undirected network, there exist deterministic local rules, following which any point-to-point demand can be routed with polylogarithmic competitiveness --- while the best known results before had completion time relative to the product of the optimal congestion and dilation, we only have polylogarithmic overhead to their sum.
    
%     This is achieved through a combination of recent results in semi-oblivious routing with novel deterministic scheduling algorithms: we provide a deterministic scheduling algorithm that is polylog-competitive when the set of packets to be scheduled is guaranteed to be a subset of a polynomial-size set of \textit{candidate paths}. Furthermore, our scheduling algorithm extends to \textit{scheduling with noise}, where a hidden set of packets, the \textit{signal}, with low congestion and dilation, is mixed with undesirable packets, the \textit{noise}. We show that a $1-o(1)$ fraction of the signal can be scheduled with polylogarithmic overhead to the signal's congestion and dilation, regardless of the congestion and dilation of the noise, even though it is impossible to discern if a packet is part of the signal or noise.
% \end{abstract}

\begin{abstract}
  \noindent This paper addresses point-to-point packet routing in undirected networks, which is the most important communication primitive in most networks. The main result proves the existence of routing tables that guarantee a polylog-competitive completion-time {\bfseries deterministically}:

  \begin{quote}
  \emph{In any undirected network it is possible to give each node simple stateless deterministic local forwarding rules, such that, any adversarially chosen set of packets are delivered as fast as possible, up to polylog factors.}
  \end{quote}
  
  \noindent All previous routing strategies crucially required randomization for both route selection and packet scheduling.% to and need to assume that packet requests are generated independently.  
  
  %Moreover, the resulting congestion and completion time guarantees of known strategies only hold in expectation or with high probability and under the  assumption that packet requests and random choices are independent\footnote{This independence assumption is questionable in practice. Examples include adversarial denial-of-service attacks which know or can adapt to observed routing behavior. Even the fact that TCP/IP resends packets if acknowledgments time-out creates adaptive feedback behaviour that naturally creates more packets for connections that are not well-served.}.

  \medskip

\noindent The core technical contribution of this paper is a new local packet scheduling result of independent interest. This scheduling strategy integrates well with recent sparse semi-oblivious path selection strategies. Such strategies deterministically select not one but several candidate paths for each packet and require a global coordinator to know all packets to adaptively select a single good path from those candidates for each packet. Of course, global knowledge of all packets is exactly what local routing tables  cannot have. Another challenge is that, even if a single path is selected for each packet, no strategy for scheduling packets along low-congestion paths that is both local and deterministic is known. Our novel scheduling strategy utilizes the fact that every semi-oblivious routing strategy uses only a small (polynomial) subset of candidate routes. It overcomes the issue of global coordination by furthermore being provably robust to adversarial noise. This avoids the issue of having to choose a single path per packet because congestion caused by ineffective candidate paths can be treated as noise.

  \medskip

\noindent  Beyond more efficient routing tables, our results can be seen as making progress on fundamental questions regarding the importance and power of randomization in network communications and distributed computing. For example, our results imply the first deterministic universally-optimal algorithms in the distributed supported-CONGEST model for many important global distributed tasks, including computing minimum spanning trees, approximate shortest paths, and part-wise aggregates.
\end{abstract}

%  local routing strategies for , a central primitive in distributed computing. In the most general setting, we give the first randomized poly-competitive routing that is robust against adaptive adversaries. When the network is known in advance, we give a deterministic and polylog-competitive local routing.

  % For any fixed network G, we design the first 
  % When the network is known in advance we give the first deterministic and polylog-competitive local strategy, resolving 
  % knowedge, we design randomized polylog-competitive strategy that is robust to an adaptive adversary.
  % We give the first deterministic and polylog-competitive local routing strategy for point-to-point communication in undirected networks. %, resolving an important open problem in distributed computing.

  % Our result, in conjunction with prior work, completes a solid justification of semi-oblivious routing's usefulness from a theoretical standpoint. This supports prior practical efforts which showed that, empirically, semi-oblivious routing exhibits near-optimal efficiency and superior robustness in network traffic engineering~[NSDI'18].

\newpage

\tableofcontents

\thispagestyle{empty}
\setcounter{page}{0}

% \newpage

% \alert{\colorbox{pink}{Paper-writing collaboration notes:}
% \begin{itemize}
% \item use ``path set'' instead of ``path system''
% \item use ``node'' instead of ``vertices''  
% \item no oblivious or semi-oblivious *routing* should be mentioned anywhere in the paper .. use "path selection" for the \textit{strategy} and "path set" for the object!!!
% \item in the paper, we mostly talk about the deterministic case. maybe we can add a point about the randomness somewhere. perhaps say
% \item Check logarithm bases! Use $\log$ for base-2 logarithm wherever possible, $\ln$ where paired with $\exp$.
% \item Use $\poly(\log n)$, not $\poly \log n$, CHECK THAT THIS IS DONE EVERYWHERE!
% \item Be consistent between permutation demands and \zeroonedemands, which do we want to use?? \antti{the reduction from zeroonedemands to general demands does not work from permutation demands to general demands due to a technicality (can't track if a packet has crossed simulated edge at the start of a path), so better to use zero-one demands}
% \item Sections 5 and 6 (main scheduling and routing sections) have not been reviewed after big changes, please read and comment
% \item (G) I removed the term ``machine-stateless'' since all our algorithms have this property and the amount of new terms was overwhelming. I just use ``stateless'' to mean ``job-stateless''
% \item (G) I removed the distinction between greedy and greedy-enabled as it is not central to the paper and might confuse people.
% \end{itemize}}

\newpage
\section{Introduction}

Point-to-point packet routing is the most important and common communication primitive in most networks and highly efficient routing, congestion control, and scheduling procedures form the backbone of the internet and essentially all modern distributed systems. It is also related to many intensely studied mathematical problems in distributed computing, scheduling, operations research, graph algorithms, network optimization, network information theory, network coding, etc.

This paper focuses on the following extremely natural and clean mathematical formulation:

\begin{quote}
  \emph{In a network, abstracted as an undirected graph, several nodes locally generate packets. Each packet specifies a destination node to be delivered to. In each synchronous time step, nodes can forward up to one packet to each neighbor. The goal is to deliver all packets as fast as possible.}
\end{quote}

In a typical setup, once a network is wired up and routers are configured (based on the network topology but without knowledge about future traffic/packet requests), routing decisions need to be extremely efficient. Indeed, for real-time communications (e.g., even just Zoom calls) the latencies across several dozens of routing hops add up, hence it is crucial that routing decisions are done in microseconds, and even consumer routers now often process millions of packets per second. % with commercial or ISP routers being order-of-magnitude faster.

To accommodate this requirement, essentially all routing strategies construct a set of \textbf{local forwarding rules} called \textbf{routing tables}. Locality here refers to the requirement that a node must make a decision based only on its own view of the network state. The responsibility of routing tables includes both \textbf{path selection}, i.e., deciding which neighbor to forward a received packet to, and \textbf{packet scheduling}, i.e., deciding which packet to prioritize and forward first if multiple packets present at a node should go to the same neighbor. The objective to minimize the time until all packets are delivered is called \textbf{completion time}.

There are many important aspects of routing algorithms. However, one of the most practically important and conceptually fundamental is \textbf{determinism}. This is because routers are simple machines that must be extremely reliable.
% Moreover, generating large amounts of high-quality randomness in a fast, reliable, cheap way is very hard at a hardware level.
Deterministic forwarding rules that are guaranteed to work well for all possible network traffic, even traffic chosen by adversaries that know or observe the router configurations, are crucial in security-relevant settings because of their provable resilience against denial-of-service attacks.
% Deterministic forwarding rules that are guaranteed to work well for all possible network traffic, even ones chosen by adversaries that know or observe the router configurations, are crucial in security-relevant settings because of their provable resilience against denial-of-service attacks.

Unfortunately, despite significant research efforts towards understanding the power of randomization in distributed computation \cite{gerla1973deterministic,panconesi1996complexity,awerbuch1996fast,lenzen2013optimal,chen2016distributed,chang2019time,chang2020deterministic,ghaffari2020network,anagnostides2021deterministic,ghaffari2023faster}, no effective deterministic routing strategy with provable guarantees on every graph has been given. We resolve this problem by proving that for every undirected network, there exist \emph{deterministic} routing tables with polylog-competitive completion time, i.e., every node can be given some local stateless forwarding rules such that routing any set of packets using these deterministic rules delivers all packets essentially as fast as possible. 

\begin{restatable}{theorem}{introRouting}\label{thm:intro-routing}  
  For every graph $G = (V, E)$ with $|E| = \poly(n)$ there exists a deterministic and stateless routing table $A_G$ that satisfies the following. Given any set of source-destination pairs $d = \{(s_i, t_i))_{i=1}^{k}$ with $k \le \poly(n)$, $A_G$ routes each packet from $s_i$ to $t_i$ simultaneously with completion time $\poly(\log n) \cdot \OPT(d)$ where $\OPT(d)$ is the offline-optimal completion time.
\end{restatable}

Here, ``stateless'' means that the algorithm is not allowed to modify the packets when forwarding them.  

The fact that routing can be done deterministically in general networks with polylog-competitive algorithms at all is very surprising. The only interesting classes of networks for which fast deterministic routing is known are hypercubes~\cite{ajtai19830,leighton1998hypercubic} (and more generally, expanders~\cite{chang2020deterministic}), and these involve notoriously complicated sorting networks and expander-decomposition algorithms.

Furthermore, one typically decomposes the routing problem into two subproblems: \textbf{path selection} and \textbf{packet scheduling}. The former asks to select a single path between the source and destination for each packet, while the latter asks to determine at which time steps packets advance across the next edge on their predetermined path. Even in isolation, these subproblems do not offer satisfactory deterministic solutions. Any deterministic and oblivious path selection on hypercubes must suffer a competitiveness ratio of at least $\tilde{\Omega}(\sqrt{n})$~\cite{KaklamanisKT91}, and the randomized version was resolved only recently~\cite{GhaffariHZ21} for completion time. Worse yet, packet scheduling has an extremely simple polylog-competitive randomized strategy~\cite{LeightonMR94}, but no non-trivial deterministic one is known.

We bypass all of these results by leveraging the recently developed \textbf{semi-oblivious path selection}~\cite{GBA23}. While deterministic, semi-oblivious path selection previously lacked a local and efficient (even randomized!) packet selection strategy, hence did not give notable theoretical guarantees for routing. We design such a scheduling strategy in this paper, derandomize it, and finally resolve the question of deterministic routing.

While our packet scheduling is especially adapted to semi-oblivious paths and is not fully general, we believe it might be of independent interest. Beyond being deterministic, it is provably resilient to adversarial noise --- a crucial property required to route semi-oblivious paths. Furthermore, it can be combined with \textbf{greedy} strategies: it can handle arbitrary opportunistic advancements ahead of schedule, a very important practical consideration that enables efficiency in low-congestion conditions~\cite{cidon1995greedy}.

\textbf{Computational aspects.} The primary goal of this paper is to prove the existence of deterministic routing tables, a combinatorial object unrelated to computational issues. However, we provide an efficient Monte Carlo construction for these routing tables where randomness is used but the constructed object is good with high probability. Moreover, if successful, the routing table is then robust against even an adaptive adversary (that can see all sampled random bits and any demand pattern indefinitely in the future. We leave the question of an efficient deterministic construction for future work.

\textbf{Consequences for the principles of distributed computing.} Point-to-point routing is a fundamental primitive in distributed computing upon which many other algorithms are built. Specifically, a long line of work in the theory of low-congestion shortcuts and universal optimality~\cite{GH16,haeupler2016near,haeupler2018minor,2022sssp,GHR21,haeupler2020network,ghaffari_haeupler2021shortcuts_in_minor_closed,ghaffari2021universally,anagnostides2021almost,goranci2022universally,haeupler2021universally} culminated with the conclusion that near-optimal point-to-point routing is the \emph{only} barrier required to achieve near-optimal distributed algorithms for many important problems including the minimum spanning tree, approximate shortest path, part-wise aggregation, etc. Moreover, a recent paper by Ghaffari and Zuzic~\cite{ghaffari2021universally} derandomized the entire framework except for the point-to-point routing; the missing piece is contributed by this paper. An important consequence of our work is that we can now design \emph{deterministic} \textbf{universally-optimal algorithms}~\cite{GKP98,haeupler2021universally} for all of the above problems: Namely, for every network $G$, we design \emph{deterministic} algorithms that are $\poly(\log n)$-competitive with the fastest possible algorithm on $G$. In terms of distributed models, our results apply to the supported-CONGEST model~\cite{schmid2013exploiting}: a standard message-passing model with $O(\log n)$-messages per time step where the topology is known in advance (but not the input). The supported-CONGEST model is appropriate for applications where the network is fixed in advance.
\begin{restatable}{theorem}{thmMainSupported}\label{thm:main-supported} 
  There exists a deterministic universally-optimal algorithm in supported-CONGEST for the following tasks: exact minimum spanning tree, $(1+\frac{1}{\poly(\log n)})$-approximate single-source shortest path, and part-wise aggregation.
\end{restatable}

We note that, while we expect our results to be useful for the distributed setting, the construction of the routing table is not distributed: the construction requires the global knowledge of the graph $G$, hence the supported-CONGEST model. Studying algorithms under this model has been increasingly popular as many applications can safely assume the network is fixed~\cite{schmid2013exploiting,KumarYYFKLLS18}: e.g., a supercomputer's interconnect topology doesn't change, data centers rarely move, a software-defined network controller knows the global topology, etc. 

Moreover, resolving the known-topology restriction might be hard as an efficient distributed algorithm for this task would circumvent the pervasive ``expander-routing barrier'' where the task is to find a distributed $\poly(\log n)$ round CONGEST algorithm for solving the point-to-point routing task in an expander. The barrier is not solved even in the randomized setting. Partial results towards resolving it include \cite{ghaffari2017distributed,GL18b,GHR21,chang2020deterministic}.

%Indeed, many barriers remain in the CONGEST model: even with randomization, we currently do not know how to obtain near-optimal universally-optimal results for all graphs.

\textbf{Consequences for principled network traffic engineering.} Recent studies in the context of network traffic engineering (TE) have explored routing strategies that use semi-oblivious path selection~\cite{kumar18semi,KumarYYFKLLS18}. Namely, they consider a small constant number of candidate paths (in their case, $4$) between each pair of nodes and exclusively send traffic along those paths. However, even after fixing the paths, this approach requires one to choose for each packet which candidate path to use and to schedule the packets along the chosen path.

In principle, these problems might require a lot of global coordination and might have very fragile solutions that need to be frequently recomputed. On the other hand, Kumar et al.~\cite{kumar18semi,KumarYYFKLLS18} show that a simple global coordinator performs well in practice. Our work partially explains this success by proving that the problems of packet scheduling and deciding between the candidates are solvable with local algorithms (i.e., without global coordination) and the solutions are extremely robust to even adversarial noise. Furthermore, our work might pave the way for future deterministic, more local, and more efficient strategies.

\textbf{Consequences for job shop scheduling.} % \alert{GZ: explain model, explain results. Basically, jobs can only move forward (no backward movement I guess). This is achivable since there is no noise.} 
%\alert{CS: I added a few paragraph about job shop scheduling here. 
% Since the FOCS reviewers questioned the model and local is probably less motivated for job shop, I don't think we should say more than what I wrote, and possibly less.  I also did not formally define the problem, as I don't think it is worth the space to do so.}
There is a close relationship between packet scheduling and the classical job shop scheduling problem (see e.g. \cite{Pinedo2022}). In job-shop scheduling, you have a set of machines, and a job consists of an ordered series of operations, where each operation has a particular processing time on a particular machine. We can cast the packet scheduling problem as a job-shop scheduling problem by calling each edge a machine and each path a job, where the operations are unit-sized and the path defines the ordered set of machines. We contribute to the development of job shop scheduling algorithms that are deterministic and local. While these properties have not been traditionally the focus of the scheduling literature, our other contributions of greedy-enabled (see \Cref{sec:scheduling-model}) algorithms that are near-optimal to the congestion+dilation are central~\cite{ShmoysSW94,GoldbergPSS01,FeigeS02,MastrolilliS11} and may have further ramifications.
%The connection to job shop may have further ramifications, since the state-of-the art approximation algorithms for job shop use, as a lower bound, the analog of congestion plus dilation (maximum machine load plus maximum job length), and techniques such as \cite{LeightonMR94, LeightonMR99} are used in some of the algorithms.   The approximations for these problems depend on the variant, and are roughly $\log$ or $\log^2$ of the lower bound ften with near-matching lower bounds.

% Randomised algorithms that only work against \textbf{oblivious} adversaries can break down when packet traffic adapts to how routing is done, for example if clients with delayed or dropped packets are re-sending them; the guarantees of deterministic algorithms hold regardless of how the packet set is chosen.

\section{Preliminaries}
\label{sec:prelims}

Routing can be decomposed into two separate subtasks: \textbf{path selection} and \textbf{packet scheduling}. The former asks to select a single path between the source and destination of each packet, while the latter asks to determine at which time steps packets advance across the next edge on their predetermined path. We note that papers often do not differentiate between the two; either subtask is often referred to as ``routing''. However, this paper draws a clear distinction between them.

\textbf{(Classic) Packet scheduling.} Suppose a set of paths $P$ to be used for routing is predetermined (e.g., computed by a previous phase of an algorithm) and our goal is to compute a schedule that minimizes completion time. This goal is closely tied to the \textbf{congestion} and \textbf{dilation} parameters, defined below.
\begin{definition}
  The length $\len(p)$ of a single path $p$ is the number of edges on that path (we consider unit graphs). The \textbf{dilation} of a path set $P$ is $D(P) := \max_{p \in P} \len(p)$, i.e., the largest length of a path in the set. The \textbf{congestion} of the path set $P$ is $C(P) := \max_{e \in E} |\{ p : P \mid e \in p \}|$, i.e., the maximum number of paths that contain any one edge in $G$.
\end{definition}

Trivially, no schedule can take fewer than $D(P)$ time steps as this is insufficient to schedule even the longest path. Similarly, ignoring congestion issues on all but the most-congested edge requires $C(P)$ time steps. Therefore, any schedule requires $\Omega(C(P) + D(P))$ time steps. An $O(C(P) + D(P))$ upper bound, matching the lower bound up to constants, was proven in the seminal paper of Leighton, Maggs, and Rao~\cite{LeightonMR94}. In other words, the optimal completion time to schedule packets over a path set $P$ is $\Theta(C(P) + D(P))$. However, achieving this bound requires extensive global coordination and is infeasible for local algorithms.

On the other hand, no similar result is known that can be applied in the deterministic and local setting. It is easy to show that arbitrary forwarding yields a schedule with completion time $C(P) \cdot D(P)$, but no deterministic local algorithm has been able to bypass even the $C(P) \cdot D(P)$ barrier by more than a polylogarithmic factor. We bypass this question by tackling both the path selection and scheduling problems together and give a deterministic end-to-end routing that is $\poly(\log n)$-competitive in terms of completion time with the optimal offline solution. %\antti{Would maybe also be nice to be able to say that path selection must be random as well before this, we bypass two problems!} \shyr{I'm a little wary about selling the path selection too hard, since we mostly use ZHR to compute these. But maybe this is fine?} \antti{that's true, I guess we are "only" bypassing the problem of requiring centralised control to select the good path subset from the semi-oblivious routing instead of the full path selection problem. This is probably good (but talking about motivation is really hard so I cannot really be sure here)}

Our approach crucially relies on two recent breakthroughs in the theory of oblivious routing: hop-constrained oblivious routing~\cite{GhaffariHZ21} and sparse semi-oblivious routings~\cite{GBA23}, both explained below. We now describe these crucial ingredients.

\textbf{Oblivious path selection.} To give some context, a path-selection strategy is called \textbf{oblivious}~\cite{gupta2006oblivious,Racke08,englert2009oblivious,Racke019} if the path assigned to each source-destination pair $(s_i, t_i)$ depends only on $s_i$ and $t_i$ (and the underlying network $G$). The path selection algorithms used most heavily by the Internet are in practice either oblivious or close to oblivious. In a seminal result, Raecke~\cite{Racke08} showed that for every network there exists an oblivious and stochastic path-selection strategy that is in expectation $O(\log n)$-competitive with the offline optimum in terms of $C(P)$ for all demands. However, this path selection produces paths of unbounded dilation hence is not appropriate for minimizing completion time.

More recently, Ghaffari, Haeupler, and Zuzic~\cite{GhaffariHZ21} designed the ``hop-constrained oblivious path selection'' which produces stochastic paths that are $\poly(\log n)$-competitive in expectation with the \emph{offline-optimum completion time}. We first formally define the term.

\begin{restatable}{definition}{defDemand}\label{def:demand}
  A \textbf{demand} $d$ is a collection $d = \{ (s_i, t_i) \}_{i=1}^k$ of $k$ node pairs, where the same pair of nodes is allowed to appear more than once. A demand is a \textbf{permutation demand} if no two pairs share a common source or sink, i.e., $s_i \neq s_j$ and $t_i \neq t_j$ when $i \neq j$. A demand is a \textbf{$\{0,1\}$-demand} if no pair appears more than once. A path set $P$ \textbf{satisfies} the demand $\{(s_i, t_i)\}_{i=1}^k$ if for every node pair $(a, b)$, if the pair appears $m$ times in $d$, there are at least $m$ $(a, b)$-paths in $P$.
\end{restatable}

\begin{definition}\label{def:offline-opt-completion-time}
  Given a demand $d := \{(s_i, t_i)\}_{i=1}^k$ in a graph $G$, the \textbf{offline-optimum completion time} is defined as $\OPT(d) := \min_{P^*} C(P^*) + D(P^*)$, where the $\min$ is taken over all path sets $P^*$ that satisfy the demand.

  A path selection or routing strategy is $\alpha$-\textbf{competitive} in terms of completion time if, for all demands, it gives a strategy $P$ with $C(P)+D(P)$ or completion time (resp.) that is at most $\alpha$ times the offline-optimum completion time.
\end{definition}
%\antti{what is a "routing solution"?} \shyr{I think we can just say "A path selection..." here. Doesn't seem like ``routing solution'' is used anywhere else.}

In more detail, given a demand $d$ and $OPT(d)$, the oblivious stochastic path selection of \cite{GhaffariHZ21} produces near-optimal paths $P$ with dilation $D(P) = \poly(\log n) \cdot OPT(d)$ and expected congestion $\E[ C(P) ] \le \poly(\log n) \cdot \OPT(d)$.

Of course, oblivious path selection strategies are inherently randomized and no deterministic structure with similar properties can exist \cite{KaklamanisKT91}. %exists (or can exist, as proven in~\cite{KaklamanisKT91}).
%Moreover, it is unclear how to design deterministic packet scheduling even if they were known.

\textbf{Semi-oblivious path selection.} Raecke's $O(\log n)$-competitive oblivious path selection matches the (existential) $\Omega(\log n)$ lower bound. Motivated by this success, a prominent extension called ``semi-oblivious path selection'' was suggested by Hajiaghayi, Kleinberg, and Leighton~\cite{HajiaghayiKL07} with the hope of surpassing the $\Omega(\log n)$ barrier.

In a nutshell, a semi-oblivious path selection is allowed to (obliviously) select \emph{multiple} \textbf{candidate paths} for each source-destination node pair. The setting guarantees the existence of a single path for each pair from the candidates such that the chosen paths are competitive with the offline optimum. However, finding these choices is a priori a hard problem requiring global coordination.

Furthermore, theoretical results for semi-oblivious path selection have been primarily negative: \cite{HajiaghayiKL07} showed that any such selection with polynomially-many candidate paths cannot be $o(\frac{\log n}{\log \log n})$-competitive in terms of edge congestion --- essentially no better than standard oblivious path selection in the worst case. Similarly, Czerner~\cite{czerner2020semi} shows similar semi-oblivious barriers in directed graphs.

The only positive result was contributed recently by Zuzic, Haeupler, and Roeyskoe~\cite{GBA23} who consider \textbf{deterministic} $\alpha$-\textbf{sparse} semi-oblivious path selections. In this setting, at most $\alpha \ge 1$ many candidate paths are selected between each source-destination pair.
% For example, setting $\alpha = 1$ reduces to the oblivious setting \alert{!!! chance this!}.
\cite{GBA23} show that every graph has a deterministic $O(\log n)$-sparse semi-oblivious path selection that is $\poly(\log n)$-competitive in terms of congestion $C(P)$. Moreover, by combining their results with hop-constrained oblivious routings~\cite{GhaffariHZ21}, they also obtain $O(\log^2 n)$-sparse semi-oblivious path selection that is $\poly(\log n)$-competitive in terms of completion time.

Of course, several challenges remain in using semi-oblivious selections for deterministic and local routing. No local strategy for selecting a path from the candidates was available prior to this paper. Furthermore, even if assuming a path set is selected, no deterministic packet scheduling strategy is known.

\textbf{Routing tables.} In this paper, a \textbf{routing table} is any local algorithm for routing.

\section{Technical contribution}
\label{sec:techoverview} 
\textbf{Our contribution: Deterministic local routing when the topology is known.}
Our main result is a deterministic and local end-to-end routing that is $\poly(\log n)$-competitive in terms of completion time with the optimal offline solution. We formalize our routing model in \Cref{sec:routing-model}.

\introRouting*

% \begin{restatable}{corollary}{deterministicroutingexists}\label{cor:deterministicroutingexists}
%   In any graph, there exists a deterministic, local stateless routing algorithm $A$, such that for all polynomially-sized demand $d$, the completion time of $A$ on $d$ is $\tildeO{(\OPT(X))}$ where $\OPT(d)$ is the offline optimum.
% \end{restatable} 

We achieve this routing result by combining the semi-oblivious path selection~\cite{GBA23} with a new robust packet scheduling strategy developed in this paper.
%In the above statement, the permutation requirement is a technical condition that can be removed in many ways, but is sufficient for our distributed computing application.

\textbf{Technical contribution: Robust scheduling on a ``domain'' path set.} Our main technical result is a novel \emph{scheduling} algorithm for semi-oblivious routing that is local and deterministic and might be of independent interest. Furthermore, our scheduling is \textbf{greedy-enabled}, which means that the algorithm can handle an arbitrary number of opportunistic packet forwards that are ahead of schedule. Our scheduling model is formalized in \Cref{sec:scheduling-model}.

% \antti{was wrong citation? switched it}% \cite{Albers05}. %And an algorithm is stateless if it requires no mutable bits in the packets' header.

While we design our scheduling algorithm to be used with semi-oblivious path selection, our scheduler requires only one inherent assumption: we require that the paths $P$ we want to schedule must be a subset of some (so-called) \textbf{domain path set} $\calP$ and we require that $\calP$ has at most polynomial size. One should think of $\calP$ as the set of all possible paths that can appear in the input. However, the actual demand $d$ can be any subset of $\calP$ and our algorithms are oblivious to the actual demand while being polylog-competitive with $\OPT(d)$ (and not $\OPT(\calP) \gg \OPT(d)$).

We develop a deterministic and randomized version of the algorithm. Our deterministic scheduler needs to know both $\calP$ and the underlying graph $G$. On the other hand, our randomized scheduler is fully general in that it requires upfront knowledge of either $\calP$ or the graph $G$; and unlike prior work is robust even against an adaptive adversary.

We emphasize that, while our scheduler requires the domain-set assumption, the final deterministic routing algorithm (i.e., \Cref{thm:intro-routing}) does not! This is because we combine our scheduler with the semi-oblivious path-selection strategy that quietly provides a polynomially-sized set of domain paths such that restricting our paths to this domain allows for polylog-competitive routing for all demands. In other words, the domain-set assumption is the crucial technical insight that is ultimately internal to our routing algorithm.

% The deterministic scheduler depends on $\calP$ and therefore requires $\calP$ to be known along with the underlying graph $G$. This assumption can be dropped if one allows for randomization in choosing the deterministic scheduler (which then still works deterministically for all demands). \shyr{Should we also mention here that our deterministic algorithm is not efficient?} \antti{the existence result one is efficient, the actually constructive one isn't, so it's fine}

However, even outside of our scheduling algorithm, the domain-path-set restriction is reasonable in many applications, such as, say, a large company choosing polynomially-many potential paths between each pair of its data centers.

% The competitiveness of our scheduling algorithm only depends on $\calP$ via a $O(\log |\calP|)$ term, therefore allowing for (polynomially) large domain path sets. 

\alert{TODO: make sure Crefs that should go to the below result actually go to it, and those to the random version go to that, and vice versa, as their labels got switched at some point.}

\begin{restatable}{theorem}{introPathScheduling}
\label{thm:intro-path-scheduling}
Let $\calP$ be a domain path set and $G = (V,E)$ a directed graph where $\calP, V, E$ have cardinality at most $\poly(n)$. There exists a deterministic and local algorithm with completion time $\tildeO(\OPT(\ps)) = \tildeO(C(\ps) + D(\ps))$ on any subset $\ps \subseteq \calP$. The algorithm can either be stateless, or can be greedy-enabled and have 1-bit states.
\end{restatable}
\begin{proof}
  This is a rephrasing of \Cref{cor:deterministicstatelessschedulingexists} and \Cref{cor:deterministicgreedyschedulingexists}, proven later.
\end{proof}

Here, ``1-bit states'' means that we require a single mutable bit in the header of a packet that the scheduling algorithm can read and write.

\textbf{Robustness to adversarial noise.} In fact, our scheduling algorithm solves a harder problem than the one above.

Suppose the set of paths to schedule $P$ can be decomposed, in a way not revealed to the algorithm, into $P = S \sqcup N$. Here, the path set $S$ represents the \textbf{signal} that can be scheduled with very low completion time of $\OPT := C(S) + D(S)$ time steps. On the other hand, the path set $N$ represents arbitrary \textbf{noise} the adversary injects. Suppose that the whole set of paths $P$ has large $\OPT(P)$, but there exists a large subset $S$ with a much smaller $\OPT(S)$. We can adapt our algorithm to schedule a constant fraction of $S$ in time competitive with $\OPT(S)$. Formally, If the signal-to-noise ratio is bounded from below by some $\beta^{-1} < 1$ (i.e., $|S| / |S \sqcup N| \ge \beta^{-1}$), then we can schedule at least a constant fraction of $S$-paths in completion time competitive with $S$, namely in $\OPT \cdot O(\beta \log n)$ time steps.

\begin{theorem}\label{thm:intro-scheduling-with-noise}
  Let $\calP$ be a domain path set and $G = (V,E)$ a directed graph where $\calP, V, E$ have cardinality at most $\poly(n)$. There exists a deterministic, greedy-enabled scheduling algorithm with 1-bit packet states that takes parameters $\beta > 1, T > 1$ and path set $P \subseteq \calP$ as input, completes in $\tilde{O}(T)$ time steps, and satisfies the following. For all subsets $S \subseteq P$ such that $|S| \geq \frac{1}{\beta}|P|$ and $\OPT(S) \le T$ we have that at least $\frac{1}{4}$-fraction of packets in $S$ are completed.
\end{theorem}
\begin{proof}This is a rephrasing of \Cref{thm:noisyschedulingexists}, proven later.\end{proof}

\textbf{Putting path selection and scheduling together.} Combining the above robust scheduling with completion-time sparse semi-oblivious path selection, we obtain our headline result. Namely, we start with the very recent semi-oblivious path selection~\cite{GhaffariHZ21,GBA23}: For every graph $G$, there exists a deterministic collection $\calP$ of $O(\log^2 n)$ candidate paths between each pair of nodes such that, for each demand $d$, the union of these paths support a solution with completion time $\poly(\log n) \cdot \OPT(d)$. More precisely, for any fixed demand (i.e., set of source-sink pairs $d = \{(s_i, t_i)\}_i$) with offline-optimum completion-time $\OPT(d)$, there is a hidden subset $S \subseteq \calP$ (called the signal) that achieves $C(S) + D(S) = \poly(\log n) \cdot \OPT(d)$. The rest of the candidate paths are noise $N := \calP \setminus S$. Furthermore, the sparsity guarantees a polylogarithmic signal-to-noise ratio, hence we can use our robust and deterministic scheduling algorithm restricted to the $O(n^2 \log^2 n)$-sized path set $\mathcal{P}$ (the domain). This schedules at least a constant fraction of the demand in $\poly(\log n) \cdot \OPT(d)$ time steps; we identify the pairs that were satisfied, remove them from the demand, and repeat $O(\log n)$ time to get the result.

\subsection{Proof Overview: Deterministic Scheduling}
We start by discussing the proof of \Cref{thm:intro-path-scheduling}. Our proof is via the probabilistic method \cite{alon2016probabilistic}. Namely, we show that there is a randomised local scheduling algorithm that with high probability is competitive against an adaptive adversary on a domain path set. Note that this implies a deterministic scheduling algorithm as some seed must be good.

For consistency with the remainder of the paper, we will change our terminology from packet scheduling the job-shop scheduling that is more general and common in the literature. In a nutshell, one simply substitutes ``packet'' $\to$ ``job'', ``edge'' $\to$ ``machine'', and ``path'' $\to$ ``sequence of machines''. In this setting (formalized in \Cref{sec:scheduling-model}), there is a set of machines $M$ and we denote the set of all machine-sequences as $\seqset$. There is a domain sequence set $\calS \subseteq \seqset$ with polynomial size, and a set of jobs $\js$, each with a sequence of machines (from $\calS$) that must work on it, in the order given by the sequence, to complete it. The jobs start on the first machines of their sequences, and each machine can work on up to one job per time step, at which point that job moves to the next machine in its sequence. The scheduling has to be done locally and without maintaining state: each time step, each machine must select which job to work on based only on the jobs it holds, the current time step, the number of machines and an upper bound on polynomially bounded quantities, and the shared randomness. The goal is to minimize the \textit{completion time}: the time until all jobs are completed.

\textbf{Idea: exponential success probability.}
For convenience, for this discussion we assume access to an upper bound $L$ on the congestion + dilation of the job set. This assumption can be eliminated by starting with constant $L$ and repeatedly doubling it.

There are simple randomised scheduling algorithms that succeed with high probability, for example, the standard random delay algorithm. This algorithm samples a uniformly random integer $h(\jo) \sim \{ 0, 1, 2, \ldots, L-1 \}$ for each job $\jo$. The job is inactive until time step $h(\jo)$, at which point it activates, and starts moving through its sequence. Up to a few minor technicalities, this approach completes in $\tilde{O}(L)$ time steps \emph{with high probability}, meaning that it succeeds with probability at least $1 - n^{-b}$ where the constant $b > 0$ can be made sufficiently large. However, with high probability guarantees are \emph{woefully insufficient} for our setting: even for polynomial $|\calS| = |M|^{\Theta(1)}$, there are $|\calS|^{\poly |M|}$ possible polynomial-size sequence sets for jobs, which is exponential in $|M|$. We get around the requirement of exponential success probability by instead requiring success probability exponential \emph{in the size of the job set}: We require that the probability the algorithm fails on a job set $\js$ of size $s := |\js|$ is at most $\exp(-s (\ln |\calS| + \Omega(\log |M|)))$. Since there are $|\calS|^{s} = \exp(s \ln |\calS|)$ possible job sequence sets of size $s$, the probability of there existing a failing set of size $s$ is $\exp(-s \cdot \Omega(\log |M|)) < 1$. Therefore, it is sufficient to design an algorithm whose success probability is exponential in the size of the job set.

\textbf{Strategy: weak scheduling.} To achieve the desired probability bounds, we focus on an easier problem, \textit{weak scheduling}, where success is defined as completing \emph{at least half} of the jobs within the time bound. We design a ``weak scheduling algorithm'' where the probability of $k$ jobs failing to be completed is roughly $\exp(-\tilde{\Omega}(k))$, which will be sufficient for our result.

\textbf{Note: tolerating arbitrary starting positions (greedy-enableness).}
Consider a simple idea to increase the probability of success: sample multiple, independent, weak scheduling algorithms and run them back-to-back. If a sampled algorithm fails on a job set with probability $p$, the probability each of $k$ independently sampled algorithms would fail on the job set is $p^k$. 

However, after a failed call to a weak scheduler, jobs have made partial progress. This changes the remaining job set! Even if one of the sampled algorithms would have succeeded on the initial job set, it might not succeed with jobs starting further along their sequences, and less than half of the jobs might be delivered during the whole process. Hence, to increase the success probability with this method, our weak scheduling algorithm's success cannot depend on the initial positions of the jobs --- we need to tolerate arbitrary starting positions.

\textbf{The local weak scheduling subroutine.} We now give details of the randomized weak scheduling subroutine. For simplicity, we assume here that no two jobs have the same sequence of jobs. It will run for $2L$ so-called \textit{large time steps}, each of which will be subdivided into $l = \bigO{\log n}$ \textit{small time steps} (for a total of $2 \cdot L \cdot l$ time steps). The small time steps effectively allow machines to work on $l$ jobs per time step instead of just one, i.e., they enable $l$ jobs to pass through a machine each large time step.

In addition to the time step lengths $L$ and $l$, the subroutine receives as input a \textit{hash function} $h$ that maps sequences in $\calS$ to integers in $\{0, \dots, L - 1\}$.

\textbf{Strategy: virtual times.} In the standard random delay algorithm, delays are handled by having jobs wait at their initial machines for that many time steps, then moving once per $\bigO{\log |M|}$ time steps (i.e., once per large time step) to simulate machines with higher capacity. However, this approach does not tolerate starting positions that are chosen by an adaptive adversary, i.e., in response to the fixed randomness. A more robust approach is needed. 

To tolerate arbitrary starting positions, and further, to achieve a greedy-enabled variant of the scheduling algorithm, we define \textit{virtual times} for jobs, which correspond to the last large time step the job $\jo$ can be at machine $\seq(\jo)_i$ while remaining "on schedule". For a job $j$, we define its virtual time for machine $\seq(\jo)_i$ as $h(\jo) + i$. 

With these definitions in mind, the weak scheduling algorithm is as follows. At machine $m$, at the start of large time step $T$, if there are strictly more than $l$ jobs with virtual time $T$ at $m$'s queue, all of them are \textit{dropped}, removed from consideration for the rest of the subroutine. Otherwise, during the large time step $T$, the machine has time to work on every job with virtual time $T$.

Now by definition, non-dropped jobs with virtual time $T$ at machine $m$ cannot arrive to $m$ during or after large time step $T$: all jobs that fall behind schedule are dropped. As every virtual time is will be bounded by $L + D(\js)$ and $L \geq D(\js)$, after $2L$ large time steps, every job is either completed or dropped. Thus, it suffices to bound the probability that more than half of the jobs are dropped.

The analysis of the success probability of the algorithm on a fixed set of jobs is done solely through virtual times. Virtual times only depend on the hash function and job sequences, not on the positions of jobs or even the greedy adversary in the greedy-enabled model.

\textbf{Idea: bad patterns.} Dropping jobs on machines introduces dependencies between the jobs. For example, a dropped job will not congest later machines in its sequence. This can potentially cause a machine with high congestion in the job sequence set to not drop packets, complicating any analytical arguments. In order to address these dependencies in an elegant way, we use a notion of \textit{bad patterns}.

As discussed earlier, we must show that the probability of failing to complete half of the jobs is \emph{exponentially small} in the size of the job set. We do this by defining \textit{bad patterns}, which describe how the virtual times defined by a hash function might cause a failure. A bad pattern is a collection of disjoint subsets $B_{T, m}$ of (job, position) pairs indexed by (large time step, machine)-pairs, such that the sets $B_{T, m}$ are either empty or have size greater than $l$, and the total size of the subsets $\sum |B_{T, m}|$ is more than $|\js| / 2$. We say a bad pattern \textit{occurs} if for every $(T, m)$, every job in $B_{T, m}$ has virtual time at machine $m$ equal to $T$.

Note that a prerequisite for a set of jobs to be dropped at machine $m$ at large time step $T$ is for all of those jobs to have virtual time $T$ at $m$. Hence, if more than half of the jobs are dropped, some bad pattern is guaranteed to occur. In particular, if $B'_{T, m}$ is the set of jobs dropped at time $T$ on machine $m$, then the bad pattern $B_{T, m} = B'_{T, m}$ occurred. This collection $B$ is indeed a bad pattern, as whenever jobs are dropped at a machine, more than $l$ are dropped, and as more than half of the jobs are dropped, $\sum |B_{T, m}| > |\js| / 2$. Thus, it suffices to bound the probability a bad pattern occurs. In other words, bad patterns in a sense overcount the number of congestions, giving us a manageable upper bound.

The probability of any fixed bad pattern occurring is at most $L^{-\sum_{T, m} |B_{T, m}|} \leq \exp(-1/2 \cdot |\js| \ln L)$, which is exponential in the size of the job set. This is because the probability of any job having a fixed virtual time at a machine is at most $L^{-1}$ since the delays are sampled from $\{0, \ldots, L-1\}$ and $\sum_{T, m} |B_{T, m}| > |\js| / 2$. Through simple combinatorial calculations, we can bound the number of bad patterns to complete the proof. 

\subsection{Proof Overview: Scheduling with Noise}
\label{sec:techoverview-noise}

We move to scheduling with noise (\Cref{thm:intro-scheduling-with-noise}). In this setting, the set of jobs $\js$ has high congestion and dilation, but there is a good subset $\js_{S}$: a large low-congestion low-dilation subset of jobs. The scheduling algorithm receives no information on what packets are part of the good subset but receives a \textit{noise level} parameter $\beta$ satisfying $|\js| \leq \beta |\js_S|$. The goal is to schedule jobs such that at least a constant fraction, namely a fourth, of $\js_S$ is delivered in time proportional to $\beta \cdot [C(\js_{S}) + D(\js_{S})]$, despite the "noise", the part of the job set causing it to have high congestion and dilation.

Since the whole job set has high optimal completion time, it is impossible to expect all jobs or even just all jobs in the good set to be delivered. Therefore, weak scheduling is the best possible result. To see this, consider the case where we take any set of jobs $\js_{S}$, single out a single job $\jo$, and pad out the whole job set with $(\beta - 1)|\js_S|$ copies of the singled out job $j$. It is impossible to differentiate between the copies of $\jo$, and only a small number $\beta \cdot [ C(\js_S) + D(\js_S) ] \ll \beta |\js_{S}|$ can be completed within the allotted time. Hence, even for a randomised algorithm, the probability of completing the singled out job is at most $[ C(\js) + D(\js) ] / |\js|$, i.e., negligible. Further, the factor $\beta$ is necessary: if every job has the same sequence, completing a 1/4-fraction of the good subset requires completing a 1/4-fraction of the entire job set, which requires $\bigO{C(\js) + D(\js)} = \bigO{\beta C(\js) + D(\js)}$ work.

For the noisy setting, our main contribution is to observe that the same weak scheduling subroutine can be modified to function in the setting, simply through small modifications in the step length $l$. Specifically, by enlarging $l' := 4 \beta l$, we show that for a good subset $\js_{S}$, noise level $\beta$ and uniformly random hash function $h$, with failure probability exponentially small in $|\js_{S}|$, the subroutine completes at least a fourth of $\js_{S}$ regardless of the noisy set $\js \supseteq \js_{S}$, as long as $|\js| \leq \beta |\js_{S}|$. This is done through analysis with bad patterns as before, but an additional idea is required.

\textbf{Idea: large and small events.} Directly applying the bad pattern analysis to scheduling with noise fails, as we are not guaranteed that a large number of jobs from the good subset are dropped every time jobs are dropped (a majority of the jobs dropped might be from the noise). We get around this by splitting events where jobs from the good subset drop into two sets: large events, those where at least $\frac{l'}{4\beta}$ jobs \emph{from the good subset} are dropped, and small events, those where less than $\frac{l'}{4\beta}$ are.

In each small event, while more than $l'$ jobs are dropped, out of those at most $\frac{l'}{4\beta}$ are from the good subset. At most $|\js| / l'$ events of this kind can happen (as at that point every job has been exhausted), thus at most $\frac{|\js|}{l'} \cdot \frac{l'}{4\beta} = \frac{|\js|}{4\beta} \leq \frac{|\js_S|}{4}$ good jobs are dropped in small events.

To bound the number of jobs dropped in large events, we again use bad patterns, this time with sets $B_{T, m} \subseteq \js_{S}$ with the requirement that $|B_{T, m}| > \frac{l'}{4\beta} = l$ instead of just $l'$, and $\sum |B_{T, m}| > \frac{|\js_{S}|}{4}$ instead of $\frac{|\js|}{2}$. The analysis reverts to the case one without noise! The small and large events together account for at most $\frac{|\js_S|}{4} + \frac{|\js_S|}{2} = \frac{3|\js_S|}{4}$ drops of jobs from the good subset.

\subsection{Combination with Semi-Oblivious Path Selection}
Finally, we combine the noisy scheduling with sparse semi-oblivious path sets to achieve \Cref{thm:intro-routing}. An $\alpha$-sparse $\gamma$-competitive semi-oblivious path set is a collection of up to $\alpha$ paths for every $(s, t)$-pair $\{R_1(s,t), R_2(s, t), \ldots, R_\alpha(s, t)\}$. The path set has the property that for any demand $d = \{(s_i, t_i)\}_{i=1}^k$ (\Cref{def:demand}), there exists a subset $P' \subseteq P$ of the paths $P := \bigcup_{i=1}^k \bigcup_{j=1}^\alpha R_j(s_i, t_i)$ sourced from the semi-oblivious path set. This $P'$ satisfies the demand $d$ and has congestion + dilation at most at most $\gamma$ times the offline optimum. Note that the graph $G$ is fixed and known to the nodes, hence we can use the $\poly(\log n)$-sparse $\poly(\log n)$-competitive semi-oblivious path set developed in the recent paper~\cite{GBA23}.

Our main contribution is conceptual: we note that our scheduling with noise nicely combines with semi-oblivious path set to yield local routing algorithms. Given a demand $d$, our strategy will be to repeatedly route at least half of the demand. To do this, we take the path set $P$ with all up to $\alpha$ $(s, t)$-paths in the semi-oblivious path set for every $(s, t)$ with nonzero demand and use noisy scheduling on this path set. By the competitiveness of the semi-oblivious path set, there exists a low-congestion low-dilation path set $S$ (the signal), and as it contains one $(s, t)$-path for every $(s, t)$ with nonzero demand, we have $|P| \leq \alpha |S|$. As the noise level is polylogarithmic and the congestion and dilation of the signal have polylogarithmic overhead to the offline optimal routing of the demand (as the semi-oblivious routing is $\poly(\log n)$-competitive), half of the packets are delivered with only polylogarithmic overhead to the offline optimum. Repeating this a logarithmic number of times routes the whole demand.

After the repetitions, we achieve strong routing (i.e., we satisfy the entire demand). This is in spite of scheduling with noise being unable to provide such guarantees. The reason behind this contrast is that each time a single packet is delivered between $s$ and $t$, all the other $\alpha - 1$ paths between $s$ and $t$ are removed from $P$.

\section{Related Work}\label{sec:related-work}

\textbf{Local packet scheduling.} Given a set of paths with congestion $C = C(P)$ and dilation $D = D(P)$, the simplest deterministic solution that forwards arbitrary packets completes in $C\cdot D$ time steps. Moreover, no deterministic local algorithm for arbitrary networks and arbitrary paths has broken the barrier of requiring $\Omega(CD / \poly(\log n))$ time steps prior to this paper. This lack of progress has spurred research into several subcases (see below), as well as lower bounds. For example, \cite{LeightonMR94} prove that any ``non-predictive'' deterministic strategy requires at least $\Omega( CD / \log C )$ time steps. Here, non-predictive means that the order of forwards across an edge $e$ does not depend on the path's route beyond $e$. Cidon, Kutten, Mansour, and Peleg~\cite{cidon1995greedy}
showed that many natural deterministic greedy strategies may require $\Omega(D \sqrt{k} + k)$ time steps, where $k$ is the number of packets.

On the upper bound for local scheduling, a simple randomized delay gives a schedule with expected completion time $O( (C + D) \log n )$ for arbitrary graphs and paths. Leighton, Maggs, Ranade, and Rao gave a randomized local $O(C + L + \log n)$-time-step scheduling for bounded-degree $L$-level network~\cite{leighton1994randomized}. A $L$-level network is one where edges are directed from level-$i$ nodes to level-$(i+1)$ nodes and the levels are within $[0, L]$. For arbitrary graphs, Ostrovsky and Rabani~\cite{ostrovsky1997universal} designed a $O(C + D + \log^{1+\eps} n)$-time local randomized algorithm.

Several works considered the case when the paths are the shortest paths in the network. Of course, this can significantly increase the congestion $C$ compared to the optimal (in terms of completion-time) paths. Mansour and Patt-Shamir~\cite{mansour1991greedy} show that (local) greedy routing on shortest paths has completion time $D+k-1$. This is often suboptimal as, in general, $C \ll k$. Meyer auf der Heide and Vocking~\cite{auf1995packet} devised a simple local randomized algorithm that routes all packets to their destinations in $O(C + D + \log kD)$ steps, provided that the packets follow shortest paths. See the nice survey of Maggs~\cite{maggs2006survey} for further references.

\textbf{Offline scheduling and path selection.} An offline algorithm has access to both the complete description of the network, as well as the entire input (the demand for routing, or paths for scheduling). Offline packet scheduling in the offline setting originated from the job scheduling literature as the graph edges can be interpreted as machines and paths as jobs that need access to different machinexs in some order. In a seminal work, Leighton, Maggs, and Rao~\cite{LeightonMR94} showed an offline algorithm for packet scheduling which, given a path set $P$, gives a scheduling algorithm with completion time $O( C(P) + D(P) )$, matching the $\Omega(C(P) + D(P))$ lower bound up to constant factors. Haeupler, Hershkowitz, and Wajc~\cite{haeupler2019near} consider rooted trees instead of paths and show that there exist schedules in $O(C + D + \log^2 n)$ time steps, as well as a lower bound saying the result cannot be improved to $O(C + D) + o(\log n)$. Ghaffari~\cite{ghaffari2015near} considered DAGs instead of paths and showed there exist examples with completion time $\Omega(C + D\log n)$ (up to $\log \log n$ factors).

For the path selection problem, Srinivasan and Teo~\cite{srinivasan1997constant} design a constant-approximation and poly-time path selection algorithm. In other words, given a demand $d$, they produce paths $P$ with completion time $C(P) + D(P) = O( \OPT(d) )$. Bertsimas and Gamarnik~\cite{bertsimas1999asymptotically} output paths with completion time $(1 + o(1))\OPT(d)$ for sufficiently large $\OPT(d)$.

%\goran{mention lenzen and patt-shamir}

\textbf{Distributed routing on fixed networks.} Hypercubes pose an attractive theoretical topology due to their simplicity and expanderness. Early parallel systems even interconnected their processors in (variants of) hypercube topologies~\cite{cypher1994hypercube}. A large amount of effort was invested into developing algorithms specialized for hypercubes. Valiant~\cite{valiant1982scheme} proposed a simple, randomized, and optimal $O(\log n)$-competitive routing algorithm for hypercubes. A deterministic variant was given by Batcher~\cite{batcher1968sorting} who designed the bitonic sorting network, which immediately gives a $O(\log^2 n)$-competitive routing algorithm in hypercubes. This was improved by the seminal result of Ajtai, Komlos, and Szemeredy~\cite{ajtai19830} who proved the existence of a $O(\log n)$-depth sorting network, proving the existence of optimal deterministic routing algorithms. Rabin~\cite{rabin1989efficient} gave fault-tolerant and constant-queue-size routing algorithms for the hypercube. For other networks, randomized oblivious path-selection algorithms in terms of completion time were developed for many specialized networks. These include meshes~\cite{busch2008optimal}, hole-free networks~\cite{busch2010optimal}, geometric networks~\cite{busch2005oblivious}, particular extensions of hypercubes called HyperX networks~\cite{ahn2009hyperx}, deterministic algorithms in expanders~\cite{chang2020deterministic}, etc.

\textbf{A comparison with~\cite{GBA23}} This paper can be seen as a continuation of the line of work on sparse semi-oblivious routing strategies started by \cite{GBA23}. While \cite{GBA23} proves the existence of a competitive and sparse semi-oblivious routings, its main drawback was the lack of an appropriate scheduling algorithm adapted to the setting that does not require extensive global coordination. Moreover, there was little reason that any local scheduling exists --- such scheduling must be able to handle adversarial paths (inserted by the semi-oblivious selection), making the problem seem much harder than the standard setting (for which no deterministic and local algorithm exists!). In spite of these challenges, this paper offers a satisfactory conclusion to this line of work by showing that end-to-end routing via the (sparse) semi-oblivious path selection is viable and even offers theoretical advantages over other routing strategies. From a technical perspective, this paper uses and expands upon the techniques used by \cite{GBA23}. Specifically, it shows that chaining $\poly(\log n)$ random delays and allowing for a constant fraction of failures will make a demand fail with probability that is exponentially small in the size of the demand. Furthermore, this paper observes that growing the overcongestion threshold by a multiplicative factor of $O(\beta)$ makes the algorithm robust to sets with adversarial noise (with signal-to-noise ratio at least $\beta^{-1}$). With small-but-important adaptations, the developed techniques of bad patterns can be readily applied to both of these analyses. Finally, this paper shows that the contributions of the two papers can be chained together in a synergistic way and applied to important open problems from distributed computing.

\section{Deterministic Local Scheduling}\label{sec:detlocsec}

In this section we contribute our scheduling algorithm. We first formally introduce the model in \Cref{sec:scheduling-model} in the job-shop terminology. Then we present our formal results in \Cref{sec:scheduling-formal-results}. We present the main building block of our scheduling (called ``weak scheduling'') in \Cref{sec:weakdetscheduling}. We then present our full scheduling algorithm in \Cref{sec:localrulescheduling}, and finally we extend the algorithms to deal with adversarial noise in \Cref{sec:noisydetsec}.

\subsection{Scheduling Model}\label{sec:scheduling-model}

Recall the packet scheduling problem, discussed in the introduction:

\textbf{(Classic) Packet scheduling.} There is a graph and a set of packets. Each packet has a fixed routing path that it must follow, and starts at the first edge of its path. Each time step, each edge can select up to one packet, which moves to the next edge on its path.

Instead of using the packet scheduling model, we present our results in the \textbf{job-shop scheduling model with unit-sized jobs}.

\begin{definition}[Job-Shop Scheduling Model With Unit-Sized Jobs]
  There is a set of machines $M$ and a set of jobs $\js$. Each job has a fixed sequence of machines that must work on it in order, and starts at the first machine of its sequence. Each time step, each machine can work on up to one job, which then moves to the next machine in its sequence. Note that in this model, we have \textbf{unit-sized jobs}, meaning every machine takes one unit of time to finish working on a given job.
\end{definition}

A job is completed when all of the machines on its sequence have worked on it in order. The goal is to minimise the \textbf{completion time}, i.e. time until every job is completed, of our algorithm $A$, denoted $\compl(A,\js)$.

Recall that this model generalizes packet scheduling, as it can emulate packet scheduling by setting edges to machines, packets to jobs, and paths to job sequences. Further, as the paths in packet scheduling are routing paths, they can often be chosen to be simple, while in the job-shop scheduling model, sequences of machines can contain the same machine multiple times to encode that a machine may need to complete multiple tasks on a particular job.

We also set up some notation for this model. Let $\seqset$ denote the set of all machine sequences. Let $\seqn(\jo) \in \seqset$ denote the sequence of machines for a job $\jo \in \js$. Additionally, we assume each job has a unique identifier $\ind(j) \in \mathbb{N}$ that is polynomial bounded in $|M|$, to allow ourselves to distinguish between jobs with the same sequence. $\pos(\jo)_t \in \{0, \dots, \len(\seqn(\jo))\}$ will denote the next machine to work on $j$ at time $t$. Additionally, jobs have a state $\state(j)_t \in \mathbb{N}$, which is initially $0$ and can be modified by the scheduling algorithm. Each machine $m \in M$ has a \textbf{queue}, $\que(m)$, containing the jobs at that machine: $\que(m)_t := \{\jo \in \js : \seqn(\jo)_{\pos(\jo)_t} = m\}$ at time step $t$.

For a job set $\js$, we denote the \textbf{congestion}, maximum number of times any machine appears in the job sequences in total, of the job set by $\cng(\js) := \max_{m \in M} \sum_{\jo \in \js} |\{i : \seq(\jo)_i = m\}|$, and the \textbf{dilation}, length of the longest job sequence by $\dil(\js) := \max_{\jo \in \js} \len(\seq(\jo))$. Additionally, all our results will assume that the machine sequences come from a \textbf{domain sequence set} $\calS \subseteq M^*$.

\textbf{Locality.} An algorithm is \textbf{local} if each machine can select the job to work on and new states for the jobs in its queue based \emph{only} on the current time step and the history of contents of the machine's queue. The full job set to be scheduled or the contents of any other machines' queues are unknown to and unusable by machines for making
decisions. In this paper, we only consider local scheduling algorithms, and further, only consider \textbf{machine-stateless} algorithms, where machines must make decision only based on the the current time step and the current contents of the machine's queue. Notably they cannot access the history of their queue.

\textbf{Statelessness.} An algorithm is \textbf{stateless} if it does not access or change jobs' states. An algorithm is said to have \textbf{$b$-bit job states} if jobs have internal states from $\{0, \dots, 2^{b} - 1\}$. In this paper, we only consider algorithms that are either (job-)stateless or have 1-bit job state.

\textbf{Greedy-Enabledness.} We say an algorithm is \textbf{greedy-enabled} if it can handle arbitrary opportunistic job advancements. Formally, we extend the the job-shop scheduling setting by introducing an adversary. Each time step, the adversary first learns the position and state of every job. Then, the adversary selects any number of jobs, moving them further along their sequences, incrementing $\pos(\jo)$ by an arbitrary nonnegative amount $\mathrm{push}(\jo)$ for each of them. Then, the scheduling algorithm performs the time step with the updated job positions as in the job-shop scheduling model. The adversary is adaptive in case of randomized algorithms. We define the completion time of the scheduling algorithm on the job set in the model as the maximum completion time over adversaries and denote it by $\greedycompl(A, \js)$. 

\textbf{Noisy Scheduling.} We will also be interested with the model with noise, which naturally arises when one tries to combine our scheduling algorithm with semi-oblivious path selection algorithms. For a set of jobs $\js$, call a subset of jobs $\js_{S}$ $(\beta, T)$-good if it is both large (it satisfies $|\js_{S}| \geq \frac{1}{\beta} |\js|$) and has low congestion and dilation (it satisfies $\cng(\js_{S}) + \dil(\js_{S}) \leq T$). We extend our scheduling algorithms to show that if there exists a $(\beta, T)$-good subset of the job set $\js$, we can complete at least a $\bigO{\beta^{-1}}$-fraction of the entire job set in roughly $\beta \cdot T$ time. In fact, we show a stronger statement: roughly $\beta \cdot T$ time suffices to complete a constant fraction of \emph{every} $(\beta, T)$-good subset.

\textbf{Discussion: Domain Sets.} Our scheduling algorithms require the existence of a \textbf{domain sequence set} $\calS \subseteq \seqset$ that the jobs' sequences are pulled from: we are guaranteed that the job set $\js$ is \textbf{$\calS$-supported}, meaning $\seq(\jo) \in \calS$ for all $\jo \in \js$. A factor of $\log |\calS|$ appears in the completion time of our scheduling algorithms; this term is logarithmic for polynomial $|\calS|$. For our deterministic scheduling algorithms, $\calS$ must be known by every machine. However, we in addition give a scheduling algorithm with shared randomness that does not require access to $\calS$ (only $|\calS|$) and works even against an adaptive adversary, that can select the ($\calS$-supported) job set based on the random bits of the algorithm.

\subsection{Formal Results}\label{sec:scheduling-formal-results}

In this section, we give the formal version of our scheduling results that we prove in the rest of \Cref{sec:detlocsec}.

%In this section, we prove two variations of the shared-randomness adaptive-adversary scheduling result, \Cref{thm:statelessschedulingexists} for a stateless algorithm and \Cref{thm:greedyschedulingexists} for a greedy-enabled $1$-bit job state algorithm.

% \begin{restatable}{theorem}{statelessschedulingexists}\label{thm:statelessschedulingexists}
% For every graph $G$ and domain set $\calP$, \Cref{alg:localrulescheduler} with access to $\calP$ is a stateless algorithm for local rule scheduling that schedules any path set $P \subseteq \calP$ with completion time $[C(P) + D(P)] \cdot (\log |\calP|) \cdot \poly(\log n)$.
% \end{restatable}

% \begin{restatable}{theorem}{greedyschedulingexists}\label{thm:greedyschedulingexists}
% For every graph $G$ and domain set $\calP$, \Cref{alg:greedylocalrulescheduler} with access to $\calP$ is a greedy-enabled algorithm for the local rule scheduling problem with $1$-bit packet states that schedules any path set $P \subseteq \calP$ with completion time $[C(P) + D(P)] \cdot (\log |\calP|) \cdot \poly(\log n)$.
% \end{restatable}

We first present the full scheduling result in the randomized setting that is local and stateless.
\begin{restatable}{lemma}{statelessschedulingexists}\label{thm:statelessschedulingexists}
  Let $M$ be a set of machines and $\calS \subseteq \seqset$ be a \textit{domain set} of size at most $|\calS| \le |M|^c, c = O(1)$. \Cref{alg:localrulescheduler} (algorithm $A$) is a randomised algorithm with shared randomness for the job-shop scheduling problem that has the following properties:

\begin{itemize}
    \item The algorithm is local and stateless. That is, each machine decides which job in its queue to work on based only on the jobs in its queue, the current time step, $|M|$, $c$, and the shared random bits.
    \item The algorithm is competitive against an adaptive adversary: with high probability over the shared random string, $A$ is competitive for all $\calS$-supported job sets: for all job sets $\js$ with polynomially bounded size, congestion and dilation that satisfy $\seqn(\jo) \in \calS$ for $\jo \in \js$, $A$ has completion time $\compl(A, \js) = \tildeO(C(\js) + D(\js))$.
\end{itemize}
\end{restatable}

We also present a version of the above theorem for the greedy-enabled setting with 1-bit states.
\begin{restatable}{lemma}{greedyschedulingexists}\label{thm:greedyschedulingexists}
Let $M$ be a set of machines and $\calS \subseteq \seqset$ be a \textit{domain set} of size at most $|\calS| \le |M|^c, c = O(1)$. \Cref{alg:greedylocalrulescheduler} (algorithm $A$) is a randomised algorithm with shared randomness for the greedy-enabled job-shop scheduling problem that has the following properties:

\begin{itemize}
    \item The algorithm is local and uses 1-bit job states. That is, each machine decides which job in its queue to work on and new states for jobs in its queue based only on the jobs in its queue, the current time step, $|M|$, $c$, and the shared random bits, and the algorithm only ever changes jobs' states to $0$ or $1$.
    \item The algorithm is competitive against an adaptive adversary: with high probability over the shared random string, $A$ is competitive for all $\calS$-supported job sets: for all job sets $\js$ with polynomially bounded size, congestion and dilation that satisfy $\seqn(\jo) \in \calS$ for $\jo \in \js$, $A$ has completion time $\greedycompl(A, \js) = \tildeO(C(\js) + D(\js))$.
\end{itemize}
\end{restatable}

Results for randomised algorithms against an adaptive adversary imply the existence of deterministic algorithms as the following trivial corollaries:%\alert{TODO: explain why these are trivial corollaries: above algorithm works with high probability, thus there exists a random string for which it works. Hardcode that random string to the algorithm, it's now a deterministic algorithm that always succeeds.}

\begin{restatable}{corollary}{deterministicstatelessschedulingexists}\label{cor:deterministicstatelessschedulingexists}
Let $M$ be a set of machines and $\calS \subseteq \seqset$ be a polynomial-size \textit{domain set} of polynomial-length sequences of machines. Then, there exists a deterministic, local, and stateless job-shop scheduling algorithm $A$ that, for all $\calS$-supported job sets $\js$ with polynomial size has completion time $\compl(A, \js) = \tildeO(C(\js) + D(\js))$.
\end{restatable}

\begin{restatable}{corollary}{deterministicgreedyschedulingexists}\label{cor:deterministicgreedyschedulingexists}
Let $M$ be a set of machines and $\calS \subseteq \seqset$ be a polynomial-size \textit{domain set} of polynomial-length sequences of machines. Then, there exists a deterministic, local, and 1-bit job state greedy-enabled job-shop scheduling algorithm $A$ that, for all $\calS$-supported job sets $\js$ with polynomial size has completion time $\greedycompl(A, \js) = \tildeO(C(\js) + D(\js))$.
\end{restatable}
% \deterministicgreedyschedulingexists*

The construction of the algorithms for the above two statements is split into two subsections. \Cref{sec:weakdetscheduling} covers the subroutines for \textit{weak scheduling}, the subproblem of completing at least half of the jobs. The subroutines take as input a \textit{hash function}, and guarantee that for a uniformly random hash function, we fail to complete at least half of the jobs with probability exponentially small in the size of the job set. Then, in \Cref{sec:localrulescheduling}, the subroutine is used to construct the algorithms.

The last subsection, \Cref{sec:noisydetsec}, adapts these algorithms for noisy scheduling. There, we prove the following result:

% \begin{restatable}{theorem}{noisyschedulingexists}\label{thm:noisyschedulingexists}
% Let $G$ be a graph, $\calP$ be a domain path set, and $\beta > 1$ the fixed \textit{noise level}. Then, \Cref{alg:localnoisyscheduler} (parameterized by $\calP$ and $\beta$) is a stateless algorithm for local rule scheduling, such that:

% For every signal $S \subseteq \calP$, for every noisy set $P \subseteq \calP$ containing the signal of small enough size ($|P| \leq \beta |S|$), the algorithm when ran on a packet set corresponding to $P$ delivers at least a $1/4$-fraction of the packets corresponding to $S$, the signal packets, in $[C(S) + D(S)] \cdot \beta \log |\calP| \cdot \poly(\log n)$ time steps.
% \end{restatable}

\begin{restatable}{lemma}{noisyschedulingexists}\label{thm:noisyschedulingexists}
Let $M$ be a set of machines, $c$ a constant (such that "polynomially bounded" means "at most $|M|^c$") and $\calS \subseteq \seqset$ be a polynomial-size \textit{domain set} of polynomial-length sequences of machines. \Cref{alg:localnoisyscheduler} is a randomised algorithm with shared randomness for the job-shop scheduling problem that has the following properties:

\begin{itemize}
    \item The algorithm is local and stateless. That is, each machine decides which job in its queue to work on based only on the jobs in its queue, the current time step, $|M|$, $c$, $\beta$ and $T$ (defined below), and the shared random bits. The algorithm does not use job states.
    \item The algorithm can competitively schedule noisy subsets against an adaptive adversary: with high probability over the shared random string, the following holds:
    
    For all $\calS$-supported job sets $\js$ with polynomial size, congestion and dilation, for all noise levels $\beta$ and polynomial completion time bounds $T$,
    \begin{itemize}
        \item The algorithm, with parameters $\beta$ and $T$, takes $\tildeO(\beta T)$ time steps
        \item For every $(\beta, T)$-good subset of jobs $\js_{S} \subseteq \js$, at least a fourth of $\js_{S}$ is completed
    \end{itemize}
    (A subset $\js_{S} \subseteq \js$ of jobs is $(\beta, T)$-good if it is both large (it satisfies $|\js_{S}| \geq \frac{1}{\beta} |\js|$) and has low congestion and dilation (it satisfies $\cng(\js_{S}) + \dil(\js_{S}) \leq T$).)
\end{itemize}
\end{restatable}

\shyr{Maybe better to use a different letter than $T$ since we use that for time step blocks in Algorihtms 1 and 2.} \antti{let's change for full version}

There is a corresponding greedy-enabled 1-bit job state version that we omit for brevity. The deterministic corollary, \Cref{cor:deterministicnoisyschedulingexists}, trivially follows.

\begin{restatable}{corollary}{deterministicnoisyschedulingexists}\label{cor:deterministicnoisyschedulingexists}
Let $M$ be a set of machines and $\calS \subseteq \seqset$ be a polynomial-size \textit{domain set} of polynomial-length sequences of machines. Then, there exists a deterministic, local, and stateless job-shop scheduling algorithm $A$, such that for all $\calS$-supported job sets $\js$ with polynomially bounded size, congestion and dilation, noise levels $\beta$ and polynomial completion time bounds $T$, the algorithm with parameters $\beta$ and $T$ on the job set $\js$ satisfies the following:
\begin{itemize}
    \item The algorithm takes $\tildeO(\beta T)$ time steps
    \item For every $(\beta, T)$-good subset of jobs $\js_{S} \subseteq \js$, at least a fourth of $\js_{S}$ is completed
\end{itemize}
\end{restatable}

\subsection{Weak Scheduling With Exponentially Small Failure Probability}\label{sec:weakdetscheduling}

In this subsection we present both a stateless and a 1-bit job state greedy-enabled \textbf{weak} scheduling subroutine. The subroutine takes as input a \textbf{hash function} $h$, which maps jobs (or more accurately, (sequence, identifier)-pairs) to integers. We show that for any fixed job set $\js$, the probability that for a uniformly random input hash function, the subroutine fails to complete at least half of the jobs, is exponentially small in $|\js|$, which later allows an union bound over possible job sets. We note a few things before stating the subroutines.

First, recall the model. The algorithm is only allowed to select which job to have a machine $m$ work on and new states for jobs based on the current time step $t$ and the jobs $\que(m)_t$ in the machine's queue. Thus, the correct form of an algorithm for the model is a function that takes as input both $m$ and $t$, and maps to up to one job in $\que(m)_t$ to work on and new states for jobs in the queue. While \Cref{alg:localrulescheduler} and \Cref{alg:greedylocalrulescheduler} are stated this way, we will state later algorithms iteratively for convenience.

The model guarantees that job identifiers are polynomial. We will in particular assume that $\ind(\jo) \leq |M|^c$ for some global constant $c$, and denote the set of job identifiers by $I := [|M|^c]$. We will additionally, without loss of generality, assume $|M| \geq 32$.

% The greedy-enabled subroutine \emph{is not greedy}. As discussed earlier, the greedy-completion time of any algorithm in the greedy-enabled model is at most the completion time of the original algorithm, so this only makes our results for the subroutine \emph{stronger}. The algorithm is stated with no job being worked on on line 10 to make the similarities between the stateless and greedy-enabled algorithm more clear. In fact, to avoid doing the same analysis twice, we show that the two algorithms are \textbf{equivalent} in the basic (non-greedy-enabled) scheduling model. 

An important property of the stateless variant is that it is \textbf{tolerant to arbitrary starting positions}: our guarantees for the subroutine hold even if jobs are not required to have initial position $0$. Note that this property trivially holds for the greedy-enabled variant, as the adversary can move jobs to any positions before the first time step. % Indeed, our guarantee that either the subroutine completes at least half of the jobs is made based only on the hash function $h$ and the set of (sequence, identifier) pairs of the job set. Thus, it holds regardless of the initial positions of the jobs (and the arbitrary decisions made inside the subroutines).

The subroutines are split into $2L$ \textit{large time steps}, each consisting of $l$ \textit{small time steps} (each of which is just a single actual time step). This split is done to simulate machines having the ability to work on $l$ jobs per time step instead of just one. The subroutine takes the large step count $L$ and small step count $l$ as parameters, and the values required to complete at least half of the jobs will be $L = \Omega(C(\js) + D(\js))$ and $l = \Omega(\log |M| / \log \log |M|)$.

Both variants use the notion of \textit{virtual times}. The subroutines are somewhat analogous to the random delay scheduling algorithm in that jobs have a initial random delay, determined by the hash function. The delay determines a schedule, where the job first waits for this delay, then starts moving one machine per large time step. As the delay is determined by the hash function that every machine has oracle access to, every machine knows when a job is scheduled to be worked on by any machine in its sequence. This time is the \textit{virtual time} of the job at the machine (for nonsimple sequences, sequence position). Virtual times can enable greedy-enabledness and arbitrary starting positions, as jobs "ahead of schedule" can be freely worked on or ignored until they fall back to being "on schedule" without affecting the analysis of the algorithm, as we'll see later.

Now, after formally defining virtual times, we are ready to state the two subroutines.

\textbf{Virtual times.} For a fixed hash function $h : \seqset \times I \mapsto \{0, \dots, L - 1\}$ (with which we abuse notation and write $h(\jo) := h(\seq(\jo), \ind(\jo))$), we define the \textit{virtual time} of a (job, sequence position) pair $(\jo, i)$ as $\virt(\jo, i) := h(\jo) + i$ if $i < \len(\seq(\jo))$ and $\infty$ if $i = \len(\seq(\jo))$, and the virtual time of a job $\jo$ at a time step $t$ as $\virt(\jo)_t := \virt(\jo, \pos(\jo)_t)$.

\begin{algorithm}[H]
  \caption{Stateless Weak Scheduling Subroutine}
  \label{alg:weakschedulingsubroutine}
  \begin{algorithmic}[1]
%    \State \textbf{Global inputs:} Every edge knows the values of the ``large'' time-step length $L$ and ``small'' time step length $l$. Each edge has access to a hash function $h$ from $\calP$ to $\{0, \dots, L - 1\}$.
%     \State \textbf{Initialization:} There is a hidden set of packets $\ps$ defined by their unique paths $P = \{\pth(\pa) : \pa \in \ps\}$ and starting positions $\pos(\pa)_0 := S(\pth(\pa))$ (the algorithm does not access the identifiers).
%    \State We define the \textit{virtual time} of a (path, edge index) pair $(p, i)$ as $\virt(p, i) := h(p) + i$ and the virtual time of a packet $\pa$ at time step $t$ as $\virt(\pa)_t := \virt(\pth(\pa), \pos(\pa)_t)$.
    \Function{StatelessWeakScheduler}{m, L, l, h, t}
    \State Let $T := \lfloor t / l \rfloor$
%    \State Let $\mathrm{rem} := l - (t - Tl)$
    \State Let $Q := \{\jo \in \que(m)_t : \virt(\jo)_t = T\}$
%    \If{$0 < |Q| \leq \mathrm{rem}$}
    \If{$0 < |Q| \leq l$}
    \State Work on an arbitrary $\jo \in Q$
    \Else
    \State Do nothing
    \EndIf
    \EndFunction
  \end{algorithmic}
\end{algorithm}

\begin{algorithm}[H]
  \caption{Greedy-Enabled Weak Scheduling Subroutine}
  \label{alg:greedyweakschedulingsubroutine}
  \begin{algorithmic}[1]
%    \State \textbf{Global inputs:} Every edge knows the values of the ``large'' time-step length $L$ and ``small'' time step length $l$. Each edge has access to a hash function $h$ from $\calP$ to $\{0, \dots, L - 1\}$.
%     \State \textbf{Initialization:} There is a hidden set of packets $\ps$ defined by their unique paths $P = \{\pth(\pa) : \pa \in \ps\}$ and starting positions $\pos(\pa)_0 := S(\pth(\pa))$ (the algorithm does not access the identifiers). Initially, each packet has state $\state(\pa)_0 = 0$.
%    \State We define the \textit{virtual time} of a (path, edge index) pair $(p, i)$ as $\virt(p, i) := h(p) + i$ and the virtual time of a packet $\pa$ at time step $t$ as $\virt(\pa)_t := \virt(\pth(\pa), \pos(\pa)_t)$.
    \Function{GreedyWeakScheduler}{m, L, l, h, t}
    \State If $t = 0$, let $\state(\jo) \leftarrow 0$ for all $\jo \in \que(m)_t$
    \State Let $T := \lfloor t / l \rfloor$
    \State Let $Q := \{\jo \in \que(m)_t : \virt(\jo)_t = T \text{ AND } \state(\jo)_t = 0\}$
    \If{$t = 0\ (\text{mod } l)$ and $|Q| > l$}
    \State Change states of jobs in $Q$: $\state(\jo)_{t + 1} \leftarrow 1$
    \EndIf
    \If{$0 < |Q| \leq l$}
    \State Work on an arbitrary $\jo \in Q$
    \Else
    \State Do nothing
    \EndIf
    \EndFunction
  \end{algorithmic}
\end{algorithm}

We start by analyzing \Cref{alg:greedyweakschedulingsubroutine}. First, we define some terminology. At time step $t$ (large time step $T = \lfloor t / l \rfloor$), we say a job $\jo$ is \textbf{eliminated} if $\state(\jo)_t = 1$ and that it is \textbf{live} if $\state(\jo)_t = 0$. We say the job is \textbf{behind schedule} if $\virt(\jo)_t < T$, \textbf{on schedule} if $\virt(\jo)_t = T$ and \textbf{ahead of schedule} if $\virt(\jo)_t > T$. Whenever the if-condition on line 5 is satisfied and $Q$ is cleared, we say the jobs in $Q$ are \textbf{dropped}. The following is an important invariant:

\begin{lemma}\label{lem:livemeansnotbehindschedule}
    In \Cref{alg:greedyweakschedulingsubroutine}, if a job $\jo$ has at any time step $t' \leq t$ been \textit{behind schedule} ($\virt(\jo)_{t'} < T$), it is eliminated ($\state(\jo)_t = 1$) at time step $t$.
\end{lemma}
\begin{proof}
    Proof by induction. The claim holds at time step $t = 0$, as all jobs start on schedule or ahead of schedule.
    
    Assume the claim holds up to time step $t$, we will show it holds for time step $t$. Firstly, note that eliminated jobs cannot become live, thus we only need to show that no live job became behind schedule (but not eliminated).

    If $t \neq 0\ (\text{mod } l)$, no job could have become behind schedule, as $T$ only increases on time steps divisible by $l$, and job positions never decrease, thus neither do their virtual times.

    Now, assume $t = Tl$ and, for the sake of contradiction, that there is some live job $\jo$ with $\virt(\jo)_t = T - 1$. Consider the time step $t' = t - l$ (large time step $T' = T - 1$). At this time step, we must have had $\pos(\jo)_{t'} = \pos(\jo)_{t}$, as otherwise we'd have $\virt(\jo)_{t'} < \virt(\jo)_{t} = T - 1 = T'$. Thus, $\jo$ was in the set $Q$ of the machine it was at, and as the job is still live and $t' = 0\ (\text{mod } l)$, we had $|Q| \leq l$.
    
    As the job didn't leave $Q$ for the large time step (only way for that to occur is if the job was worked on, but its position had not increased), the algorithm prioritises having the machine work on jobs in $Q$, and had $l \geq |Q|$ time steps to work on jobs in $Q$ between the time steps $t'$ and $t$, that must mean a job $\jo'$ entered $Q$ at some time $t' < t'' < t$. However, then $\virt(\jo')_{t'' - 1} = \virt(\jo')_{t''} - 1 = T' - 1$, thus $\jo'$ was behind schedule at time $t'' - 1$, and by induction cannot be live.
\end{proof}

By the above lemma, live jobs cannot fall behind schedule. As the virtual times of jobs are either $\infty$ (if completed) or at most $2L - 2$ (assuming $L \geq \dil(J)$), after $2Ll$ time steps, any live job is complete.

\begin{lemma}
    In \Cref{alg:greedyweakschedulingsubroutine}, if $\dil(\js) \leq L$, after $2Ll$ time steps, every live job is complete.
\end{lemma}
\begin{proof}
We have $\virt(\jo)_{t} = h(\jo) + \pos(\jo)_t \leq (L - 1) + \pos(\jo)_t$. By \Cref{lem:livemeansnotbehindschedule}, if a job $\jo$ has not been completed, we must have $2L \leq \virt(\jo)_{2Ll} \leq (L - 1) + \pos(\jo)_t \leq 2(L - 1)$, a contradiction. Thus, the job must have been completed (thus $\virt(\jo)_{2Ll} = \infty$). 
\end{proof}

Thus, to show at least half of the jobs are completed in $2Ll$ time steps, it suffices to show at most half of the jobs are dropped during the algorithm. We can now introduce \textbf{bad patterns}, the tool we use to bound the probability of more than half the jobs being dropped. By the above lemma, this bounds the probability less than half of the jobs are completed (when $L \geq D(J)$).

Consider any run of the algorithm, even with arbitrary initial job positions, where more than half of the jobs were dropped. We can define sets $B_{T, m}$ of (job, position index) pairs describing what was dropped and where: $(\jo, i) \in B_{T, m}$ if and only if the job $\jo$ was dropped at machine $m = \seq(\jo)_i$ at the start of large time step $T = t / l$. The sets $B_{T, m}$ satisfy the following five properties:
\begin{itemize}
    \item Dropped jobs have virtual time equal to the large time step: $\virt(\jo, i) = T$ for all $(\jo, i) \in B_{T, m}$.
    \item Jobs cannot be dropped at machines they are not at: if $(\jo, i) \in B_{T, m}$, then $\seq(\jo)_i = m$.
    \item As each job is dropped at most once, no job $\jo$ occurs multiple times over the sets $B_{T, m}$.
    \item Whenever jobs are dropped, strictly more than $l$ of them are, so either $|B_{T, m}| = 0$ or $|B_{T, m}| > l$.
    \item Lastly, as more than half of the jobs are dropped, $\sum_{T, m} |B_{T, m}| > |\js| / 2$.
\end{itemize}
We call a collection of sets $B_{T, m}$ with the last four of the above properties a \textbf{bad pattern}, and say the bad pattern \textbf{occurs} if the first property holds. Now, more than half of the jobs can be dropped \emph{only if} some bad pattern occurs.

\begin{restatable}{definition}{badpatterndefinition}(Bad Patterns)\label{def:badpattern}
For fixed $M$, $L, l$ and $\js$, a \textit{bad pattern} is a collection of sets $B_{T, m}$ of (job, sequence position) pairs for $0 \leq T < 2L$ and $m \in M$, such that
    \begin{itemize}
        \item $\seq(\jo)_i = m$ for all $(\jo, i) \in B_{T, m}$,
        \item any $\jo \in \js$ appears at most once in $\sqcup B_{T, m}$, and
        \item the set sizes satisfy $|B_{T, m}| \in \{0\} \cup (l, |\js|]$ and $\sum |B_{T, m}| > |\js| / 2$.
    \end{itemize}
    For a hash function $h : \seqset \times I \mapsto \{0, \dots, L - 1\}$, we say a bad pattern $B$ \textit{occurs} if for all $(T, m)$ and $(\jo, i) \in B_{T, m}$, we have $T = \virt(\jo, i) := h(\jo) + i$.
\end{restatable}

When no bad pattern occurs with a hash function $h$, we call the hash function \textbf{good}:

\begin{restatable}{definition}{goodhash}(Good Hash Function)\label{def:goodhash}
For fixed $L$ and $l$, a hash function $h$ is $(L, l)$-\textit{good} for a job set $\js$ if no bad pattern occurs for $L, l, \js$ and $h$.
\end{restatable}

We now have a simple condition for \Cref{alg:greedyweakschedulingsubroutine} to complete at least half of the jobs:

\begin{lemma}\label{lem:badpatterncertifierlemma}
For fixed $M$, $L$, $l$ and $\js$ with $L \geq D(\js)$, let $h$ be a good hash function for $\js$. Then, for any decisions on which eligible jobs to work on and for any greedy-enabled adversary, \Cref{alg:greedyweakschedulingsubroutine} on $\js$ with $h$ completes at least half of the jobs in $2Ll$ time steps.
\end{lemma}

\begin{proof}
By the argument above, no bad pattern occurring certifies at least half of the jobs are completed. Since the hash function is good, no bad pattern occurs.
\end{proof}

Before continuing, we quickly get a variant of \Cref{lem:badpatterncertifierlemma} for the stateless scheduling algorithm. To do this, we note that against a greedy-enabled adversary restricted not to push any jobs at any point after the first time step, the two algorithms are \textbf{equivalent}. Thus, in particular, results for \Cref{alg:greedyweakschedulingsubroutine} against a greedy-enabled adversary hold for \Cref{alg:weakschedulingsubroutine} in the non-greedy-enabled model with arbitrary job starting positions.

\begin{lemma}\label{lem:schedulingequivalency}
Against a greedy-enabled adversary restricted to not push any jobs at any point after the first time step, \Cref{alg:weakschedulingsubroutine} and \Cref{alg:greedyweakschedulingsubroutine} are equivalent.
\end{lemma}

\begin{proof}
    Remove the IF on lines 5 to 6 in \Cref{alg:greedyweakschedulingsubroutine}. Now, note that if the size of $Q$ is greater than $l$ at the start of a large time step, its size will never decrease during the large time step, thus $|Q| \leq l$ is never satisfied and no job in $Q$ is worked on until the end of the large time step. At this point, all those jobs fall behind schedule. The algorithm only ever works on on-schedule jobs, and without the greedy adversary to push jobs between time steps, behind-schedule jobs can never become on-schedule or ahead-of-schedule. Thus, with the restriction on the greedy-enabled adversary, \Cref{alg:greedyweakschedulingsubroutine} works exactly the same without the IF on lines 5 to 6 as with it. Without the IF, the two algorithms are trivially equivalent.
\end{proof}

\begin{lemma}\label{lem:statelessbadpatterncertifierlemma}
For fixed $M$, $L$, $l$ and $\js$ with $L \geq D(\js)$, let $h$ be a good hash function for $\js$. Then, for any starting positions and decisions on which eligible jobs to work on, \Cref{alg:weakschedulingsubroutine} on $\js$ with $h$ completes at least half of the jobs in $2Ll$ time steps.
\end{lemma}

\begin{proof}
    Immediate corollary of \Cref{lem:badpatterncertifierlemma} and \Cref{lem:schedulingequivalency}.
\end{proof}

We return from the sidetrack. It remains to bound the probability a random hash function \emph{is not good} for a fixed job set $\js$. \Cref{lem:weakroutingprob} shows that for a large time step count $L$ linear in congestion + dilation and small time step count logarithmic in $|M|$, a random hash function is not good with probability \emph{exponentially small} in $-|\js|$. As discussed in the technical overview, probability this tight is required to later apply an union bound over all possible $\calS$-supported job sets.

\begin{restatable}{lemma}{weakroutingprob}\label{lem:weakroutingprob}
Let $M$ be a set of machines and $\js$ be a polynomial-size set of jobs ($|\js| \leq |M|^c$). Let $L$ and $l$ (the large and small time step counts) be integers satisfying
\begin{align*}
    2|M|^c \geq L &\geq C(\js) + D(\js),\\
    l &\geq 150c \ln |M| / \ln \ln |M|. 
\end{align*}
Then, the probability a random hash function $h$ from $\seqset \times I$ to $\{0, \dots, L - 1\}$ \emph{is not good} for $\js$ is at most $\exp\left(-|\js| \ln (l) / 8\right)$.
\end{restatable}

\begin{proof}
We need to show that the probability a bad pattern occurs is at most $\exp(-|\js| \ln(l) / 8)$.

We can easily upper bound the probability a fixed bad pattern $B$ of size $s = \sum |B_{T, m}|$ occurs by $L^{-s}$:
\begin{itemize}
    \item the probability $T = \virt(\jo, i) := h(\jo) + i$ for fixed $\jo, i, T$ and random $h$ is at most $L^{-1}$ (as $h(\jo)$ is uniformly random over $\{0, \dots, L - 1\}$), and
    \item the values $h(\jo)$ are independent of each other, thus virtual times for distinct $\jo$ are independent.
\end{itemize}

Next, we need to upper bound the number of bad patterns of any fixed size $s = \sum |B_{T, m}|$. Here we use the fact that by the definition of congestion, for any $m$ there are at most $C(\js)$ pairs $(\jo, i)$ with $\seq(\jo)_i = m$. The upper bound, \Cref{lem:badpatterncount}, is then proven by simply summing over ways to select the set sizes $s_{T, m} = |B_{T, m}|$ and multiplying this by the number of ways to select the contents of each individual $B_{T, m}$ from the up to $C(\js)$ options. The proof of the lemma is deferred to \Cref{sec:deferred}.

\begin{restatable}{lemma}{badpatterncount}\label{lem:badpatterncount}
The number of bad patterns of size $s = \sum |B_{T, m}|$ is at most
\begin{equation*}
    2\exp\left(\frac{|\js| \ln l}{8}\right) \left(\frac{eC(\js)}{l}\right)^s.
\end{equation*}
\end{restatable}

Finishing the proof simply involves summing over $s$ and multiplying the probability of a fixed bad pattern occurring with their number. We have $L \geq C(\js) + D(\js) \geq C(\js)$, thus as $l \geq 150 \geq e^5 \approx 148.4$ and $s > |\js| / 2$, we have
\begin{equation*}
    \left(\frac{eC(\js)}{l}\right)^s L^{-s} \leq (l/e)^{-s} \leq l^{-4s/5} = \exp(-(4s/5) \ln l) \leq \exp\left(-\frac{2|\js| \ln l}{5}\right).
\end{equation*}
and, as $1/8 - 2/5 = -11 / 40$, we get
\begin{equation*}
    \sum_{s = \lceil |\js| / 2 \rceil}^{|\js|} 2\exp\left(\frac{|\js| \ln l}{8}\right) \left(\frac{eC(\js)}{l}\right)^s L^{-s} \leq \sum_{s} 2\exp\left(-\frac{11 |\js| \ln l}{40}\right) \leq \exp\left(-\frac{|\js| \ln l}{8}\right)
\end{equation*}
as desired.
\end{proof}

\subsection{Scheduling Algorithm}\label{sec:localrulescheduling}

\Cref{lem:weakroutingprob} gives a probability bound for at least half of the jobs being completed for a fixed set of jobs. We want to apply an union bound over all possible job sets, to show that for a random hash function, with high probability, we can complete at least half of \emph{any} possible job set. However, this is not feasible: there are $|\calS|^{s} \binom{|I|}{s} \leq \exp(s \ln(|\calS||I|))$ possible job sets of size $s$, and the probability that a random hash function guarantees completion of at least half the jobs of a fixed job set of size $s$ is $\exp(-s \ln(l) / 8)$. As the subroutine takes $2Ll$ time steps, directly applying this strategy would lead to a completion time linear in $|\calS| \cdot |I|$, which is not desirable.

Instead, we exploit the subroutines' tolerance to arbitrary starting positions, and for each scale $L$, sample a random set $h_1, \dots, h_k$ of \emph{multiple} hash functions. The algorithm then calls the subroutine consecutively with each hash function. If at least one of the hash functions $h_i$ is good for $\js$ and $L \geq D(\js)$, at least half of the jobs are completed during a subroutine call with that hash function. Notably, this only works because the subroutines are tolerant to arbitrary starting positions. Without that guarantee, after running the weak scheduler for the first time jobs would no longer be on the first edges of their paths, and the preconditions of the subroutine would be violated, thus there would be no guarantees on successfully completing the jobs.

By \Cref{lem:weakroutingprob}, the probability a hash function is not good for a fixed job set $\js$ of size $s$ is at most $\exp(-s \ln(l) / 8)$, thus the probability none of $k$ independent random hash functions is good for $\js$ is at most $\exp(-sk \ln(l) / 8)$. Even with just $k = 8 \ln |\calS| / \ln(l)$, a polylogarithmic amount for polynomial $|\calS|$, this already equals $\exp(-s \ln |\calS|)$. Slightly higher $k$ gives high probability to succeed.

\begin{restatable}{definition}{goodhashset}(Good Set of Hash Functions)\label{def:goodhashset}
For fixed $L$ and $l$, we call a set of hash functions $h_1, \dots, h_k$ $(L, l)$-\textit{good} for a domain sequence set $\calS$ if, for every $\calS$-supported job set $\js$ with $L \geq C(\js) + D(\js)$, at least one $h_i$ is $(L, l)$-good for $\js$.
\end{restatable}

\begin{lemma}\label{lem:iterategoodprob} 
Let $\calS$ be a fixed domain sequence set, and $L \leq 2|M|^c$ and $l \geq 150c \ln |M| / \ln \ln |M|$ be fixed parameters. Then, for $k \geq 8(b + 1) \frac{\ln |\calS| + (c + 1)\ln |M|}{\ln l}$ (for some constant $b$ that will determine the probability of success), the probability a set of $k$ independently generated random hash functions mapping (sequence, identifier) pairs in $\seqset \times I$ to $\{0, \dots, L - 1\}$ is $(L, l)$-good for $\calS$ is at least $1 - 2|M|^{-b}$.
\end{lemma}

\begin{proof}
By \Cref{lem:weakroutingprob}, the probability none of $k$ independent uniform hash functions is good for a fixed $\calS$-supported job set $\js$ is at most $\exp(-\frac{|\js| k}{8} \ln l)$.

For a size $s$, there are at most $|\calS|^{s} \binom{|I|}{s} \leq \exp(s \ln(|\calS| |I|)) = \exp(s (\ln |\calS| + \ln |I|))$ ways to select a $\calS$-supported job set $\js$ with $|\js| = s$. Thus, by a union bound, the probability there exists a set $\js$ for which no $h_i$ is good is at most
\begin{equation*}
    \sum_{s} \exp(s (\ln |\calS| + \ln |I| - \frac{k}{8} \ln l))
\end{equation*}
Since $k \geq 8(b + 1) \frac{\ln |\calS| + (c + 1)\ln |M|}{\ln l} \geq 8(b + 1) \frac{\ln |\calS| + \ln |I| + \ln |M|}{\ln l}$, the exponential is at most $\exp(-bs \ln |M|) \leq |M|^{-bs}$, and as $|M| \geq 32$, the sum is at most $2|M|^{-b}$. Thus, the random set of hash functions is good with probability at least $1 - 2|M|^{-b}$, as desired.
\end{proof}

By the tolerance to arbitrary starting positions, with $L \geq C(\js) + D(\js)$, we can iterate a weak scheduling subroutine a number of times for multiple hash functions to complete at least half of the jobs. Repeating this $\lceil \log |M|^c \rceil$ times (recall the job set has polynomially bounded size $|\js| \leq |M|^c$), all the jobs are guaranteed to be completed. To deal with the dependency on $L$, we start with $L = 4$ and repeatedly double it. Note that the only difference between the stateless and greedy-enabled scheduling algorithms is the subroutine called. The algorithms are stated iteratively purely for convenience, and the constant $b$ controls the success probability in "with high probability".

\alert{TODO: figure out how to make the algorithms look good (no multiline text)}

\begin{algorithm}[h]
  \caption{Stateless Scheduling Algorithm}
  \label{alg:localrulescheduler}
  \begin{algorithmic}[1]
    \Function{StatelessScheduler}{m, $|M|$, $c$}
    \State Let $l := \lceil 150c \ln |M| / \ln \ln |M| \rceil$
    \State Let $k := \lceil 8(b + 1) (2c + 1) \ln |M| / \ln l \rceil$ \Comment{$b$ controls success probability}
    \For{$L \in \{4, 8, \dots, 2^{\lceil \log 2|M|^c \rceil}\}$}
    \State Sample a set of $k$ independent, uniformly random hash functions $h_1, \dots, h_k$ (from $\seqset \times I$ to $\{0, \dots, L - 1\}$) using the shared randomness
    \For{$i \in \{1, \dots, c \lceil \log |M| \rceil + 1\}$}
    \For{$j \in \{1, \dots, k\}$}
    \For{$t \in \{0, \dots, 2Ll - 1\}$}
    \State call $\mathrm{StatelessWeakScheduler}(m, L, l, h_j, t)$
    \EndFor
    \EndFor
    \EndFor
    \EndFor
    \EndFunction
  \end{algorithmic}
\end{algorithm}

\begin{algorithm}[h]
  \caption{Greedy-Enabled Local Rule Scheduling Algorithm}
  \label{alg:greedylocalrulescheduler}
  \begin{algorithmic}[1]
    \Function{GreedyScheduler}{m, $|M|$, $c$}
    \State Let $l := \lceil 150c \ln |M| / \ln \ln |M| \rceil$
    \State Let $k := \lceil 8(b + 1)(2c + 1) \ln |M| / \ln l \rceil$ \Comment{$b$ controls success probability}
    \For{$L \in \{4, 8, \dots, 2^{\lceil \log 2|M|^c \rceil}\}$}
    \State Sample a set of $k$ independent, uniformly random hash functions $h_1, \dots, h_k$ (from $\seqset \times I$ to $\{0, \dots, L - 1\}$) using the shared randomness
    \For{$i \in \{1, \dots, c \lceil \log |M| \rceil + 1\}$}
    \For{$j \in \{1, \dots, k\}$}
    \For{$t \in \{0, \dots, 2Ll - 1\}$}
    \State call $\mathrm{GreedyWeakScheduler}(m, L, l, h_j, t)$
    \EndFor
    \EndFor
    \EndFor
    \EndFor
    \EndFunction
  \end{algorithmic}
\end{algorithm}

% \shyr{I think it would be good to define $c$ in the algorithm for clarity. Also if we can it seems good to try and have this stated before the next section (right now there's a page break thing happening.)} \antti{Yeah, we should be more clear on what $c$ is: the constant defining "polynomially bounded".}

\statelessschedulingexists*

\begin{proof}
The algorithm is indeed local and stateless (as is the subroutine). We'll show that if every sampled set of hash functions (for every $L$) is $(L, l)$-good for $\calS$, it is guaranteed that the algorithm has the desired completion time on \textbf{every} $\calS$-supported job set, and that every sampled set of hash functions is good with high probability. Thus, the algorithm is competitive against an adaptive adversary.

By \Cref{lem:iterategoodprob}, at line 6, the probability the set of $k$ hash functions is $(L, l)$-good for $\calS$ is at least $1 - 2|M|^{-b}$ (as $|\calS| \leq |M|^c$). As there are $\bigO{\log |M|}$ different $L$, every sampled set of hash functions is good with high probability.

Now, assume every sampled set of hash functions is $(L, l)$-good for $\calS$. Fix any $\calS$-supported job set with polynomially bounded size, congestion and dilation (all at most $|M|^c$). Let $L$ be the minimum power of two that satisfies $L \geq C(\js) + D(\js)$. Since the set of hash functions $h_1, \dots, h_k$ sampled for this $L$ on line 5 is $(L, l)$-good for $\calS$, by definition for every $\calS$-supported job set $\js'$ with $L \geq C(\js') + D(\js')$, at least one $h_i$ is good. By \Cref{lem:statelessbadpatterncertifierlemma}, the subroutine with that $h_i$ is guaranteed to complete at least half of $\js'$, regardless of the jobs' starting positions. Thus, each iteration of the loop on line 7 completes at least half of the remaining job set (as the congestion and dilation of the remaining set cannot increase, thus stays below $L$, any subset of a $\calS$-supported job set is $\calS$-supported, and the weak scheduling subroutine supports arbitrary starting positions). As the total size of the job set is at most $|M|^c$, completing half of it $c \lceil \log |M|^c \rceil + 1$ times completes all of it.

As $L$ doubles every iteration of the loop on line 4, the completion time of the scheduling algorithm is dominated by the iteration that completes the job set. This iteration satisfies $L \leq 2(C(\js) + D(\js))$. Thus, assuming every sampled set of hash functions is good, the algorithm on $\js$ has completion time
\begin{equation*}
    \bigO{\log |M| \cdot k \cdot L \cdot l} = L \cdot \poly(c) \cdot \poly(\log(|M|)) = \tildeO(C(\js) + D(\js))
\end{equation*}
as desired.
\end{proof}

The proof of \Cref{thm:greedyschedulingexists} is exactly the same as that of \Cref{thm:statelessschedulingexists}, except we use \Cref{lem:badpatterncertifierlemma} to certify a good hash function guarantees half the jobs are completed, and is omitted. The algorithm has the desired completion time in the greedy-enabled model, as the only claim we care about is that a good hash function results in half of the job set being delivered, regardless of adversary, job starting positions and eligible decisions made inside the weak scheduling algorithm, which is exactly the statement of \Cref{lem:badpatterncertifierlemma}.

% \greedyschedulingexists*

% \begin{proof}
%    In the proof of \Cref{thm:statelessschedulingexists}, only the calls to the weak scheduler where $L \geq C(P) + D(P)$ were considered. When $L$ satisfies this, forwarding arbitrary packets with nonzero state for $L \geq C(P)$ time steps and setting their state to $0$ guarantees every packet is either delivered or has state $0$ afterwards. Thus, the preconditions for the greedy weak scheduler are satisfied. Due to tolerance to arbitrary starting positions, forwarding packets like this does not require modification to the analysis, and the rest of the proof is exactly like the proof of \Cref{thm:statelessschedulingexists}.
% \end{proof}

\subsection{Scheduling with Noise}\label{sec:noisydetsec}

% \antti{do we need to recap here what scheduling with noise is? I think not, we have already done it three times (introduction, formal model, technical overview), and the theorem is explicit about what scheduling with noise means}

Scheduling with noise might at first appear to be significantly harder than regular scheduling. However, remarkably, our weak scheduling subroutines \Cref{alg:weakschedulingsubroutine} and \Cref{alg:greedyweakschedulingsubroutine} extend to scheduling with noise with a very simple modification: to handle a noise level of $\beta$, just scale the small time step count $l$ by $4\beta$. Furthermore, the analysis only requires a small change: to analyse the probability more than three-fourths of a $(\beta, L)$-good subset $\js_{S}$ of jobs is dropped, we divide the job drops involving jobs in the subset to two groups: those where a small number of good jobs were dropped (at most a fourth of the good jobs were dropped in these events in total) and those where a large number were dropped (which we analyze with bad patterns, as before, reapplying the earlier results on good hash functions). The algorithm is presented iteratively for convenience only.

\begin{algorithm}[h]
  \caption{Stateless Noisy Scheduling Algorithm}
  \label{alg:localnoisyscheduler}
  \begin{algorithmic}[1]
    \Function{StatelessNoisyScheduler}{m, $|M|$, $c$, $\beta$, T}
    \State Let $L$ be the minimum power of two at least $T$
    \State Let $l := \lceil 150c \ln |M| / \ln \ln |M| \rceil$
    \State Let $l' := 4\beta l$
    \State Let $k := \lceil 8(b + 1)(2c + 1) \ln |M| / \ln l \rceil$ \Comment{$b$ controls success probability}
    \State Sample a set of $k$ independent, uniformly random hash functions $h_1, \dots, h_k$ (from $\seqset \times I$ to $\{0, \dots, L - 1\}$) using the shared randomness
    \For{$j \in \{1, \dots, k\}$}
    \For{$t \in \{0, \dots, 2Ll' - 1\}$}
    \State call $\mathrm{StatelessWeakScheduler}(m, L, l', h_j, t)$
    \EndFor
    \EndFor
    \EndFunction
  \end{algorithmic}
\end{algorithm}

\begin{lemma}\label{lem:scaleupstillgood}
    Let $\js$ be a $\calS$-supported job set and $\js_{S}$ a $(\beta, T)$-good subset of $\js$. Let the large and small time step counts $L \geq T$ and $l$ be fixed. Let $h$ be a $(L, l)$-good hash function for $\js_{S}$, and let $l' := 4\beta l$. Then, \Cref{alg:weakschedulingsubroutine} completes at least a fourth of $\js_{S}$ in $2Ll'$ time steps, when applied with parameters $L, l', h$ to $\js$. This holds regardless of the jobs' starting positions and decisions made on which eligible jobs to work on.
\end{lemma}

\begin{proof}
    Since $L \geq T \geq C(\js_S) + D(\js_S)$ as $\js_{S}$ is $(\beta, T)$-good, any jobs in the subset not dropped (defining dropped as in the greedy-enabled algorithm, that is equivalent to the stateless algorithm in the basic (non-greedy-enabled) scheduling model) will be completed. Thus, if less than a fourth of the jobs in $\js_{S}$ are completed, more than three fourths were dropped.
    
    Now, assume the contrary of the lemma: that it is possible that more than three fourths of the jobs in $\js_S$ are dropped during the $2Ll'$ time steps. We divide the job drop events into two groups: those where at most $l$ jobs in $\js_S$ were dropped (small events), and those where strictly more than $l$ were dropped (large events).

    First, note that the total number of jobs in $\js_S$ dropped in small events is at most $|\js| / 4$: every small event, at least $l' = 4\beta l$ jobs are dropped, of which at most $l$ are in $\js_S$, but at most $|\js| \leq \beta |\js_S|$ jobs may be dropped during the algorithm (as no job can be dropped more than once). Thus, small events result in at most
    \begin{equation*}
        \frac{|\js| l}{l'} \leq \frac{\beta |\js_S| l}{4 \beta l} = \frac{|\js_S|}{4}
    \end{equation*}
    jobs in $\js_S$ being dropped, as desired.

    Now, consider the large events. Define the sets $B_{T, m}$ such that if a job $\jo \in \js_S$ is dropped at position $i = \pos(\jo)_t$ at the start of large time step $T = t / l$ as part of a large event, then $(\jo, i) \in B_{T, m}$. We claim this $B$ is a bad pattern that occurs with $h$, contradicting $h$ being $(L, l)$-good for $S$. Indeed,
    \begin{itemize}
        \item $\seq(\js)_i = m$ for all $(\jo, i) \in B_{T, m}$ as jobs can only be dropped at the machine they are at,
        \item any $\jo \in \js_S$ occurs at most once in $\sqcup B_{T, m}$ as no job is dropped more than once,
        \item the set sizes satisfy $|B_{T, m}| \in \{0\} \cup (l, |\js_S|]$ and $\sum |B_{T, m}| > |\js_S| / 2$, as more than 3/4 of jobs in $\js_S$ are dropped and small events cause at most 1/4 of those jobs to drop
        \item $\virt(\jo, i) = T$ for all $(\jo, i) \in B_{T, m}$, as that is a requirement for a job to drop with position $\pos(\jo)_t = i$ during large time step $T$ at machine $\seq(\jo)_i = m$.
    \end{itemize}
    Assuming more than three fourths of the jobs in $\js_S$ were dropped leads to a contradiction, thus the subroutine is guaranteed to complete at least a fourth of the jobs in $\js_{S}$ in $2Ll'$ time steps, as desired.
\end{proof}

\noisyschedulingexists*

\begin{proof}
The algorithm is indeed local and stateless (as the subroutine is). We'll show that if every sampled set of hash functions (for every scale $L$) is $(L, l)$-good for $\calS$, it is guaranteed that the algorithm satisfies the desired property for \textbf{every} $\calS$-supported job set $\js$, and that every sampled set of hash functions is good with high probability. Thus, the algorithm is competitive against an adaptive adversary.

By \Cref{lem:iterategoodprob}, at line 6, the probability the set of $k$ hash functions is $(L, l)$-good for $\calS$ is at least $1 - 2|M|^{-b}$ (as $\log |\calS| \leq c \log |M|$). As there are $\bigO{\log |M|}$ different $L$, every sampled set of hash functions is good with high probability.

Now, assume for every $L$ the sampled set of hash functions is $(L, l)$-good. Take any $\calS$-supported job set $J$, noise level $\beta$ and polynomial completion time bound $T$, and let $L$ be the minimum power of two at least $T$. For every $(\beta, T)$-good job subset, there is then a hash function $h_i$ that is $(L, l)$-good for $J$. By \Cref{lem:scaleupstillgood}, the subroutine call with that $h_i$ is guaranteed to complete at least a fourth of the jobs in the subset, regardless of the starting positions of the jobs and decisions made on which eligible jobs to work on. Thus, calling the subroutine consecutively with each of the hash functions guarantees at least a fourth of \textbf{every} $(\beta, T)$-good subset of $\js$ is completed. Finally, this is clearly achieved in $\beta L \cdot \poly \log |M| = \tildeO(\beta T)$ time steps, as desired.
\end{proof}

A greedy-enabled version of the algorithm works and can be proven to work the exact same way, with only the subroutine call being replaced with one to \Cref{alg:greedyweakschedulingsubroutine}.

\alert{TODO: write something more about the scaling.}

Note that power-of-two scaling on $\beta L$, then on $\beta$ can remove the dependence on the two input parameters, giving completion time $\tildeO(\beta L)$. The result is stated without scaling, as this is the version of noisy scheduling required for our routing result.

\section{Local Routing through Semi-Oblivious Path Sets}\label{sec:routing-semi-obl}

\subsection{Routing Model}\label{sec:routing-model}

We define graphs and paths as usual, except that for convenience, instead of undirected graphs, we work with \textbf{reciprocal} graphs. Instead of an undirected edge, reciprocal graphs contain one edge in each direction.

\textbf{Graphs.} We denote graphs by $G = (V, E)$ where $V$ is the set of nodes and $E$ the set of directed edges. We require graphs to be \textbf{reciprocal}: each edge corresponds to an edge in the other direction. Formally, each edge $e = (u, v)$ is directed from its \textbf{head} $u$ to its \textbf{tail} $v$, and has a corresponding edge $e' = \rev(e) = (v, u)$ in the other direction (with $\rev(\rev(e)) = e$). The edge set is allowed to contain parallel edges and self-loops, but has to have polynomially bounded size.

\textbf{Paths.} A path $p$ in $G$ is a sequence of directed edges $( e_0, e_1, e_2, \ldots, e_{\len(p) - 1} )$ that are adjacent in $G$, i.e., the tail of $e_i$ matches the head of $e_{i+1}$. A path is called \textbf{simple} if no edge appears twice: $e_i \neq e_j$ for all $i \neq j$. The $i$th edge of path $p$ is denoted by $p_i = e_i$. % We denote the \textbf{length}, the number of edges in the path $p$, by $\len(p)$. 

With this, we can define the routing model:

\begin{definition}[Routing Model]
  In the routing model on a reciprocal graph $G = (V,E)$, an algorithm (with knowledge of $G$) is given a set of packets, each with a source and destination. Each edge can pass one packet over each time step, and the goal is to deliver every packet to its destination as quickly as possible.
\end{definition}

We again define some notation for this model: For an $\pa \in \ps$, let $\src(\pa)$ and $\tar(\pa)$ denote the source and destination of $\pa$. To help determinstic algorithms differentiate between packets with the same source and destination, we assume each packet also has a unique identifier $\ind(\pa) \in \mathcal{N}$ that is polynomially bounded in $n$. We let $\pos(\pa)_t \in V$ denote the position of the packet at time $t$. Additionally, every node $v \in V$ has a \textbf{queue}, $\que(v)$, that contains the packets at that node: $\que(v)_t := \{\pa \in \ps : \pos(\pa)_t = v\}$ at time step $t$. Finally, we denote the offline optimal completion time on an instance $\ps$ by $\OPT(\ps)$ and the completion time of a particular algorithm $A$ by $\compl(A,\ps)$.

As with scheduling, we are only interested in \textbf{local} routing algorithms, and further work in an even more restrictive, \textbf{local stateless} routing model: each node must decide which packet to forward over its adjacent edges based \emph{only} on the current time step and the packets in its queue.

\textbf{Known Graph Topology.} Note that in our routing model the only communication between nodes or edges is done via packet forwarding (unlike, say CONGEST or LOCAL). In such a model, one \emph{must} assume that the topology is known (or have some similar assumption). If no additional assumption is made about the graph topology or the packet set, efficient routing is not possible: set the underlying graph topology to be a tree of depth $2$ with branching factor $\sqrt{n}$, and consider a packet set with $\sqrt{n}$ packets: for each of the root's $\sqrt{n}$ subtrees, there is a packet to a leaf in that subtree. Offline optimal routing delivers the packets in $2$ time steps, however, any routing algorithm in CONGEST must take $\tilde{\Omega}(\sqrt{n})$ time steps to deliver all the packets, as to send the correct packet to a subtree, the intersection of the packets' destination node identifier set and the node identifier set of the subtree must be computed; this is a harder problem than even computing if the sets are disjoint, which has communication complexity of $\Omega(k)$ bits (when both sets have size $k$) \cite{BabaiFS86commcomplexity}.

Because of this issue, it is typically assumed in the literature that either $(1)$ the graph topology is known, or $(2)$ for every packet with source $s$ and destination $t$, there is a packet with source $t$ and destination $t$. In our model, we make the former assumption --- that all nodes know the graph topology. Of course, one can allow arbitrary coordination between nodes like in the CONGEST model, and we leave such questions to future work.

\textbf{Demands.} While the routing model does not involve demands, a demand naturally induces an instance of the routing problem, with a packet $\pa_i$ with source $\src(\pa_i) := s_i$ destination $\tar(\pa_i) := t_i$ and some polynomially bounded identifier (assumed to be worst case) for every pair $(s_i, t_i) \in d$. We write $\compl(A, d)$ to mean the completion time of an algorithm $A$ on an instance defined this way, and by $\OPT(d)$ the offline optimum completion time of such an instance. We use demands to represent instances in some parts of the paper for notational simplicity.

\subsection{Proof Overview: Combining Local Scheduling with Semi-Oblivious Path Selection}
Finally, we combine the noisy scheduling with sparse semi-oblivious path sets to achieve \Cref{thm:routingexists}. An $\alpha$-sparse $\gamma$-competitive semi-oblivious path set is a collection of up to $\alpha$ paths for every $(s, t)$-pair $\{R_1(s,t), R_2(s, t), \ldots, R_\alpha(s, t)\}$. The path set has the property that for any demand $d = \{(s_i, t_i)\}_{i=1}^k$ (\Cref{def:demand}), there exists a subset $P' \subseteq P$ of the paths $P := \bigcup_{i=1}^k \bigcup_{j=1}^\alpha R_j(s_i, t_i)$ sourced from the semi-oblivious path set. This $P'$ satisfies the demand $d$ and has congestion + dilation at most at most $\gamma$ times the offline optimum. Note that the graph $G$ is fixed and known to the nodes, hence we can use the $\poly(\log n)$-sparse $\poly(\log n)$-competitive semi-oblivious path set developed in the recent paper~\cite{GBA23}.

Our main contribution is conceptual: we note that our scheduling with noise nicely combines with semi-oblivious path set to yield local routing algorithms. Given a demand $d$, our strategy will be to repeatedly route at least half of the demand. To do this, we take the path set $P$ with all up to $\alpha$ $(s, t)$-paths in the semi-oblivious path set for every $(s, t)$ with nonzero demand and use noisy scheduling on this path set. By the competitiveness of the semi-oblivious path set, there exists a low-congestion low-dilation path set $S$ (the signal), and as it contains one $(s, t)$-path for every $(s, t)$ with nonzero demand, we have $|P| \leq \alpha |S|$. As the noise level is polylogarithmic and the congestion and dilation of the signal have polylogarithmic overhead to the offline optimal routing of the demand (as the semi-oblivious routing is $\poly(\log n)$-competitive), half of the packets are delivered with only polylogarithmic overhead to the offline optimum. Repeating this a logarithmic number of times routes the whole demand.

After the repetitions, we achieve strong routing (i.e., we satisfy the entire demand). This is in spite of scheduling with noise being unable to provide such guarantees. The reason behind this contrast is that each time a single packet is delivered between $s$ and $t$, all the other $\alpha - 1$ paths between $s$ and $t$ are removed from $P$.

% We use the paths in this path set as the domain set for noisy scheduling --- this domain set has polynomial size, thus noisy scheduling has only polylogarithmic overhead.

\subsection{Formal Results} 

In this section, we show the following statement:

\begin{restatable}{lemma}{routingexists}\label{thm:routingexists}
Let $G = (V,E)$ be a graph and $\ps$ a set of packets. \Cref{alg:routingwrapper} is a randomised algorithm with shared randomness for the routing problem with the following properties:

\begin{itemize}
    \item The algorithm is local and stateless. That is, each node decides which packets to forward based only on the packets in its queue, the current time step, the graph topology and the shared random bits.
    \item The algorithm can competitively route any packet set $\ps$ against against an adaptive adversary: with high probability over the shared random string, the algorithm is competitive on all packet sets $\ps$ i.e. $\compl(A,\ps) = \tildeO(\OPT(X))$
\end{itemize}
\end{restatable}  

By combining this result with the semi-oblivious path-selection strategy, there exists a good set of seeds for which the above algorithm always suceeds:

\introRouting*
\begin{proof}The claim follows directly from \Cref{thm:routingexists} by fixing the random bits for which the statement is true for all demands.\end{proof}

% \deterministicroutingexists

% \detroutingexists*

We prove \Cref{thm:routingexists} through a combination of noisy scheduling with semi-oblivious path sets. A path set is simply a collection of sets of paths between each node pair in the graph. A semi-oblivious path set is a path set with additional competitiveness properties. In addition, we are interested in \textit{sparse} path sets, where the number of paths between any pair of nodes is upper bounded by a parameter $\alpha$.

\begin{definition}[Path Set]\label{def:path-system}
A \textit{path set} $\mathcal{P} = \{P(s, t)\}_{s, t \in V}$ is a collection of sets $P(s, t)$ of simple paths with endpoints $s$ and $t$, for every node pair $(s, t)$. We say a path set $\mathcal{P}$ is \textit{$\alpha$-sparse} if $|P(s, t)| \leq \alpha$ for all $(s, t)$.
\end{definition}

\begin{definition}[Semi-Oblivious Path Set]
An $\alpha$-sparse $\gamma$-competitive \textit{semi-oblivious path set} $\calP$ is an $\alpha$-sparse path set, such that for every \zeroonedemand $d$, $\calP$ contains routing paths with congestion plus dilation at most $\gamma$ times the optimum: there exists a set of paths $S$ containing exactly one path from $P(s, t)$ per $(s, t) \in d$ such that $C(S) + D(S) \leq \gamma \OPT(d)$.
\end{definition}

Note that semi-oblivious path sets can be defined more generally, but in this paper, we are only interested in semi-oblivious path sets that are integral and completion-time competitive on \zeroonedemands. The above definition reflects this, omitting the quantifiers.

The recent sparse semi-oblivious path sets of \cite{GBA23} achieve polylogarithmic $\alpha$ and $\gamma$. Further, the existence is shown through proving a sampling not only is such a semi-oblivious path set with nonzero probability, but with high probability:
\begin{lemma}[Lemma 2.8 of \cite{GBA23}]\label{lem:bga23lemma28}
On every graph $G$ with $n$ nodes and a polynomially-bounded number of edges, there exists a $\bigO{\log^2 n}$-sparse $\poly(\log n)$-competitive semi-oblivious path set. Moreover, there exists an efficient, randomized algorithm that outputs such a routing with high probability.
\end{lemma}

With access to a sparse and competitive semi-oblivious path set, we can route \zeroonedemands competitively:
\begin{restatable}{proposition}{noisyschedulesemiobv}\label{thm:noisyschedulesemiobv}
    Let $G$ be a graph and $\calP$ be a fixed $\alpha$-sparse $\gamma$-competitive semi-oblivious path set on $G$. Then, \Cref{alg:semioblrouting} (parameterized by $\calP$) is a randomised local stateless routing algorithm with shared randomness that is $\tildeO(\alpha^2 \gamma)$-competitive on all \zeroonedemands against an adaptive adversary: with high probability over the shared random string, $\compl(A, d) \leq \OPT(d) \cdot \alpha^2 \gamma \cdot \poly(\log(n))$ for all \zeroonedemands d.
\end{restatable}

Combining \Cref{thm:noisyschedulesemiobv} with \Cref{lem:bga23lemma28}, we achieve deterministic polylog-overhead routing for \zeroonedemands. \Cref{alg:semioblrouting} and \Cref{thm:noisyschedulesemiobv} are covered in the next subsection, but before that, we still need to generalise to arbitrary demands. This is done through an application of the pair to a simulated, modified graph, where any arbitrary demand in the original graph maps to a \zeroonedemand.

\begin{algorithm}[h]
    \caption{Local stateless routing algorithm}
    \label{alg:routingwrapper}
    \begin{algorithmic}[1]
%    \State \textbf{Global Inputs:} $\calP$ is an $\alpha$-sparse semi-oblivious path set of node-simple paths
%    \State \textbf{Abuse of Model:} We state the algorithm iteratively for convenience only, it is memoryless.
    \Function{Route}{$v, G = (V,E)$}
    \State Create a virtual graph $G'$ by copying $G$ and for each node $b \in V(G)$ adding $n^c$ nodes $b^1, ..., b^{n^c}$ connected to $b$ with a single edge, where $n^c$ is a bound on the size of packet identifiers
    \State Sample a semi-oblivious routing $\calP'$ of $G'$ using the shared randomness via \Cref{lem:bga23lemma28}
    \State Call $\mathrm{SemiOblRouter(v, \calP')}$, where the original packet $\pa$ maps to the virtual packet $\pa'$ with source $\src(\pa)$, destination $\tar(\pa)^{\ind(\pa)}$ and identifier $\ind(\pa)$.
    \EndFunction
    \end{algorithmic}
\end{algorithm}

\begin{proof} (of \Cref{thm:routingexists}).
Note that we can assume that the paths in $\calP'$ are simple and thus only use edges in $G$ with the exception of their final edge. Moreover, by the guarantee of Lemma \ref{lem:bga23lemma28}, we have that $\calP'$ is a $\bigO{\log^2 n}$-sparse, $\poly(\log(n))$-competitive semi-oblivious path set with high probability. Since the transformation of packets maps any demand to a $\{0,1\}$ demand in $G'$, it follows that with high probability, \Cref{alg:semioblrouting} successfully routes any demand $d$ with completion time at most $\OPT_{G'}(d) \cdot \poly(\log(|(V(G')|))$. But we are now done as the optimal completion time in $G'$ is one larger than that in $G$ and the number of nodes in $G'$ is only polynomial larger than the number in $G$.

% We define a graph $G'$ that is like $G$, except with additional, simulated destination nodes: for every node $b$, there are $n^c$ corresponding simulated nodes $t_{b, i}$, each connected to $b$ with a single edge. A packet $\pa$ with source $a = \src(\pa)$, sink $b = \tar(\pa)$ and unique identifier $i = \ind(\pa)$ is then mapped to a virtual packet $\pa'$ with source $a$, sink $t_{b, i}$ and identifier $i$. Now, any possible demand is mapped to a \zeroonedemand in the virtual graph.

% We can simulate algorithm \Cref{alg:semioblrouting} with these packets and a $\poly(\log n)$-sparse $\poly(\log n)$-competitive semi-oblivious routing $\calP'$ on $G'$, except that a simulated packet $\pa'$ corresponding to packet $\pa$ with destination $b$ won't be delivered all the way to $t_{b, i}$, instead already being delivered when reaching $b$.

% The optimal completion time on the modified graph is exactly $1$ higher than the optimal completion time on the original graph. The number of nodes is only polynomially higher, thus polylogarithmic factors remain polylogarithmic, and the resulting algorithm is polylogarithmically competitive.
\end{proof}

\subsection{The Routing Algorithm}

% \goran{change the demand to be a set of pairs .. so we can use always the same definition 1.2}

In this subsection, we present \Cref{alg:semioblrouting} and prove \Cref{thm:noisyschedulesemiobv}.

Before describing the algorithm, consider first a more relaxed model, where packets can be duplicated and have a state with a logarithmic number of modifiable bits. The goal is to deliver at least one clone of each packet.

We have a \zeroonedemand $d$, and for every $(a, b)$-pair in the demand, a packet with source $a$ and destination $b$ that starts at $a$. We will accomplish our delivery goal by repeatedly delivering a constant fraction of packets for which no clone has been delivered yet. Combining an $\alpha$-sparse $\gamma$-competitive semi-oblivious route selection $\calP$ with noisy scheduling, this can be done as follows:
\begin{itemize}
    \item Initially, each node holds the packets that have it as the source. For a packet $\pa$ from $a$ to $b$, $a$ creates $\alpha$ copies of $\pa$, one for each $(a, b)$-path in the semi-oblivious route selection, labeling the copies with integers $1$ to $\alpha$. The copy with label $i$ is assigned path $P(a, b)_i$. Let $P$ be the set of all paths assigned to packets.

    The $\gamma$-competitive semi-oblivious routing guarantees that for any \zeroonedemand $d$, there exists a set $S \subseteq P$ of paths in $\calP$ that routes $d$ and has congestion and dilation at most $\gamma$ times that of the optimal routing of $D$ i.e. $[C(S) + D(S)] \leq \gamma \OPT(d)$. This set $S$ is unknown to individual nodes, but we only need to know that it exists: noisy scheduling guarantees at least a fourth of \emph{any} low-congestion low-dilation subset of large enough size gets through. Since we created a packet for \emph{every} $(a, b)$-path for every source-sink pair $(a, b)$ with nonzero demand, there is guaranteed to be a clone with the path in $S$.
    
    Recall that our scheduling algorithms required knowledge of a \textit{candidate path set}. As we know that every packet has a path from the semi-oblivious path set, we can pass the union of paths in the path set for the candidate path set. This path set has size at most $\alpha n^2$, which is polynomial.
    \item Noisy scheduling is run on the $\alpha |d|$ packets we created, with the assigned paths. Since $S$ contains one path for every node pair in the demand, and we created $\alpha$ copies of each packet, the noise level is $\alpha$. Thus, in $[C(S) + D(S)] \cdot \alpha \cdot \poly \log n \leq \alpha \gamma \cdot \OPT(d) \cdot \poly \log n$ time steps, for at least a fourth of the packets, at least one clone arrived.
    \item While a constant fraction of packets have had a clone arrive, the source nodes don't know this yet. We can apply our stateless scheduling algorithm to send feedback packets from the destinations of delivered packets to their sources, along the same paths the delivered packets used. This path set is guaranteed to be low congestion and dilation, so stateless scheduling may be used, and all feedback packets get delivered in $\alpha \gamma \cdot \OPT(d) \cdot \poly \log n$ time steps. When the source of a packet receives a feedback version of that packet, it writes to its original version of that packet the "delivered" state, at which point that packet no longer affects the algorithm. A constant fraction of packets have been delivered and eliminated!
\end{itemize}

To adapt the above to the more restrictive model where packets cannot be cloned and have no state, instead of cloning the packets, we first apply noisy scheduling to packets with their first candidate paths, then to packets not yet delivered with their second candidate paths, and so on, until all candidate paths have been attempted.

Let $P_i$ be the set of $i$th paths $P(a, b)_i$ between node pairs $(a, b)$ in the remaining demand, and $S_i = P_i \cap S$ the $i$th paths in the signal. For the $i$ such that $2 \alpha |S_i| \geq |P_i|$, the noise level is at most $2\alpha$ and the congestion and dilation of the partial signal $S_i$ are at most those of $S$, thus at most $\gamma \cdot \OPT(d)$: for these $i$, noisy scheduling on $P_i$ can efficiently deliver at least a fourth of the packets in $S_i$. The total size of the sets $S_j$ for which $2 \alpha |S_j| < |P_j|$ is at most $\sum_j |P_j| / (2 \alpha) = |S| / 2$, thus the total size of those with $2 \alpha |S_i| \geq |P_i|$ is at least $|S| / 2$, and at least $|S| / 8$ packets are delivered. Notably, there are no issues with packets being delivered before the time step their path is in the signal thanks to the tolerance to arbitrary starting positions of the algorithm (starting delivered!).

The above analysis doesn't consider one detail: packets that fail to be delivered during noisy scheduling don't magically appear back at their source node. This is resolved through applying \textit{scheduling with return}, scheduling where packets not delivered during weak scheduling are guaranteed to return to their starting node. No new ideas are required to modify \Cref{alg:localrulescheduler} to support scheduling with return in the routing model, and using this modified subroutine in the noisy scheduling algorithm has that subroutine trivially support scheduling with return too. Handling scheduling with return is deferred to \Cref{sec:returnscheduling}, the result of which is \Cref{alg:returnnoisyscheduler} with signature $\mathrm{ReturnNoisyScheduler}(v, \beta, \calP, L)$. \Cref{lem:returnnoisyschedulingworks} shows the subroutine satisfies the claimed properties:

\begin{restatable*}{lemma}{returnnoisyschedulingworks}\label{lem:returnnoisyschedulingworks}
For a fixed graph $G$, a domain set $\calP$ of node-simple paths and $\beta > 1$, \Cref{alg:returnnoisyscheduler} is a local stateless, randomized algorithm with shared randomness, such that with high probability over the shared random string, the following holds:

Let $\ps$ be any packet set with each packet having a path $\pth(\pa) \in \calP$ computable from its source, destination, unique identifier, the current time step and the graph, s.t.
\begin{itemize}
    \item no two packets evaluate to the same path, %there is exactly one packet with path $p$ for every $p \in \calP$,
    \item what path a packet evaluates to does not change for the entire duration of the algorithm, and
    \item every packet starts either at the first or last node of its path.
\end{itemize}
Denote by $\pth(\ps') := \{\pth(\pa) : \pa \in \ps'\}$ the paths assigned to a subset $\ps' \subseteq \ps$ and by $\cng(\ps') := \cng(\pth(\ps'))$ and by $\dil(\ps') := \dil(\pth(\ps'))$ the congestion and dilation of that path set. Running \Cref{alg:returnnoisyscheduler} with any parameters $\beta$, $\calP$ and $T$ on $\ps$,
\begin{itemize}
    \item the algorithm takes $\tildeO(\beta T)$ time steps,
    \item for every $(\beta, T)$-good (satisfying $|\ps| \leq \frac{1}{\beta} |\ps_S|$ and $T \geq \cng(\ps_S) + \dil(\ps_S)$) subset of packets $\ps_{S} \subseteq \ps$, at least a fourth of the packets in $\ps_{S}$ are delivered, and
    \item at the end of the algorithm, each packet is either at the first or last node of its path.
\end{itemize}
\end{restatable*}

% \shyr{Is there a nice way to simplify the statement of this lemma? Also do we want to use the ``good set'' notatation here that we do for jobs?}

The requirement for only node-simple paths in $\calP$ is required, as the scheduling model requires packets to know where along their path they are. For simple paths, this can be uniquely identified from the node the packet is at, but for nonsimple paths, this would require storing additional information in the packet, which is not allowed in the routing model. The quality of any path set is not degraded by simplifying its paths by cutting out loops, thus this requirement is not restrictive.

% Finally, we don't want to scale on $L$ (the upper bound on optimal completion time) both inside the noisy scheduling with return -subroutine and the routing algorithm itself, thus the subroutine takes $L$ as a parameter. 

% \shyr{I'm not sure we need the above explanation of why the algorithm takes $L$ as a parameter.}

\begin{algorithm}[h]
    \caption{Semi-obliviously-assisted local stateless routing algorithm}
    \label{alg:semioblrouting}
    \begin{algorithmic}[1]
%    \State \textbf{Global Inputs:} $\calP$ is an $\alpha$-sparse semi-oblivious path set of node-simple paths
%    \State \textbf{Abuse of Model:} We state the algorithm iteratively for convenience only, it is memoryless.
    \Function{SemiOblRouter}{v, $\calP$}
    \State Let $L = 4$.
    \State Let $\calP' := \bigcup_{s, t} P(s, t)$
    \While{ $L \leq n^{10}$}
    \For{$i \in \{0, \dots, 18 \lceil \log n \rceil\}$}
    \For{$j \in \{1, \dots, \alpha\}$}
    \State Call $\mathrm{ReturnNoisyScheduler}(v, 2\alpha, \calP', L)$ with $\pth(\pa) := P(\src(\pa), \tar(\pa))_j$
    \EndFor
    \EndFor
    \State Multiply $L$ by $2$.
    \EndWhile
    \EndFunction
    \end{algorithmic}
\end{algorithm}

\noisyschedulesemiobv*

% In the proof, we'll assume that $\alpha^2 \gamma \geq n^{10}$. Note that if this is not the case then routing is trivial as routing every pair $(s_i, t_i)$ greedily along the shortest $s_i, t_i$ path in $G$ delivers the packets in time at most $O(n^3)$ as these paths have congestion at most $n^2$ since there are at most $n^2$ pairs in $d$ and dilation at most $n$ as they are shortest paths.

\begin{proof}
Assume that the ReturnNoisyScheduler (\Cref{alg:returnnoisyscheduler}) succeeds every time it is called i.e. we achieve the guarantees from Lemma \ref{lem:returnnoisyschedulingworks}. This holds with high probability even against an adaptive adversary. Now fix a \zeroonedemand $d$, and let $L_0$ be the minimum power of two such that $L_0 \geq \gamma \OPT(d)$. Note that $L_0 \leq n^6$, so at some iteration in our algorithm $L$ takes value $L_0$. Since the noisy scheduler returns undelivered packets, at the start of that iteration some packets have been delivered, and all other packets $\pa$ have $\pos(\pa) = \src(\pa)$. Let $d'$ denote the demand induced by the undelivered packets, $P$ denote the set of $(s, t)$-paths for $(s, t) \in d'$ in $\calP$, and $P_j$ be the set of $j$th paths in $P$, that is, $P_j := \bigcup_{(s, t) \in d'} P(s, t)_j$.

Since $\calP$ is $\gamma$-competitive, there exists a set of paths $S$ on $\calP$ that route $d'$ with congestion + dilation at most $\gamma$ more than the optimum: $C(S) + D(S) \leq \gamma \OPT(d) \leq L_0$. Let $S_j$ be the set of $j$th paths in $S$: $S_j = S \cap P_j$.

We divide the indices into those with a large fraction of signal packets, and those with few: let $J$ be the set of indices such that $2 \alpha |S_j| \geq |P_j|$. We have
\begin{equation*}
    \sum_{j \not \in J} |S_j| \leq \sum_{j \not\in J} |P_j| / (2 \alpha) \leq \sum_{j} |P_j| / (2 \alpha) \leq |S| / 2
\end{equation*}
thus at least half of the signal packets are in indices with a large fraction of signal packets.

Now, consider a call to $\mathrm{ReturnNoisyScheduler}$ on line 7 for $j \in J$. Since $C(S_j) + D(S_j) \leq C(S) + D(S) \leq L$, by \Cref{lem:returnnoisyschedulingworks} and the assumption, the noisy scheduling subroutine delivers at least a fourth of the packets with paths in $S_j$.

Thus, the $\alpha$ iterations of the loop on line 6 deliver at least one eight of the remaining demand. As the maximum size of a unit demand is $n^2$ and $(7/8)^{18 \log n} \leq 2^{-3 \log n} = n^{-3}$, after $18 \lceil \log n \rceil$ iterations, all packets have been delivered.

Again by \Cref{lem:returnnoisyschedulingworks} and the assumption, the time spent in a iteration of the while-loop grows linearly with $L$, so only the last iteration affects the asymptotics. For that $L \leq L_0 \leq 2 \gamma \OPT(d)$, the iteration takes $\bigO{\log n} \cdot \alpha \cdot \alpha L \cdot \poly \log n$ time steps, which is $\OPT(d) \cdot \alpha^2 \gamma \cdot \poly(\log n)$.
\end{proof}

\section{Deterministic Universal Optimality in Supported-CONGEST}\label{sec:univ-opt-congest}
\shyr{TODO: Update theorems here to the new ones}

In this section we use the routing primitive developed in \Cref{sec:routing-semi-obl} to obtain strong algorithms in the supported-CONGEST model. The model is formally defined in \Cref{sec:supported-congest-prelims} and intuitively supports standard message-passing with small messages each time step, while having access to the global topology of the network (but not the input!) The rest of the section is dedicated to proving the following result.
%\Cref{thm:main-supported}: deterministic universally-optimal algorithms for many important distributed problems in the supported-CONGEST model, defined in
\thmMainSupported*

\subsection{Section-Specific Preliminaries}\label{sec:supported-congest-prelims}
\textbf{CONGEST and supported-CONGEST models.} We start with the standard message-passing model of distributed computing, often referred to as the CONGEST model~\cite{peleg2000distributed}. The network is abstracted as an $n$-node connected \emph{undirected} graph $G = (V, E)$ where each node represents one of the computers in the network (i.e., has its own processor and private memory). Communication takes place in synchronous time steps.  Per time step, each node can send one $O(\log n)$-bit message to each of its neighbors. The nodes do not know the topology of the network at the start of the algorithm. Each node only knows the unique $O(\log n)$-bit ID of itself and of its neighbors. The supported-CONGEST model inherits everything from the CONGEST model, with the addition that all nodes initially know $G$. We denote the (hop) diameter of $G$ as $D_G$.

\textbf{Tasks and distributed inputs.} As a general rule, nodes initially do not know the entire input but only their own part; at termination, the nodes should output their own part of the output~\cite{haeupler2021universally}. For example, in the minimum spanning tree (MST) problem, each node initially knows the weights of edges adjacent to it; at termination, each node must learn the weight of the MST and which of its adjacent edges belong to the MST. For the shortest path problem, all nodes know the ID of the source node and the weights of adjacent edges; at termination, each node must learn the (approximate) distance to the source and the adjacent nodes of some (approximate) shortest path tree. We note that different reasonable representations of inputs and outputs are generally equivalent.

\textbf{Part-wise aggregation.} This problem was originally defined by Ghaffari and Haeupler\cite{GH16}. Since then, it became a central primitive for many problems in distributed computing. In this problem, nodes are partitioned into \emph{connected and node-disjoint} \textbf{parts} $V = V_1 \sqcup V_2 \sqcup \ldots \sqcup V_k$. Each node $v$ initially knows the ID $i$ of its part $v \in P_i$. Additionally, each node $v$ is initially given a \textbf{private input} $x_v$. The goal is that each node computes some simple pre-defined aggregate function $\bigoplus$ (e.g., sum or max) of all private values in the same part as $v$.

\textbf{Universal optimality.} In the supported-CONGEST model and for a network $G$, an algorithm $A$ is universally optimal if the worst-case running time of $A$ across all inputs on $G$ is $\poly(\log n)$-competitive with the worst-case running time of any other correct algorithm running on $G$.

\subsection{Proof of \Cref{thm:main-supported}}

Our main tool will be \Cref{thm:intro-routing} which shows that we can solve the point-to-point permutation routing problem optimally. As an aside, we note that this solution is even stronger than the universal optimality defined above, which is a worst-case notion over all inputs. We can solve the routing problem in an \emph{instance optimal way}: our algorithm is optimal for every demand $d$.

%USE from above \semiobvrouting*
%\Cref{thm:semiobvrouting} 

\newcommand{\SQ}{\mathrm{SQ}}

\begin{lemma}
  There exists a deterministic and universally-optimal algorithm for part-wise aggregation in supported-CONGEST.  
\end{lemma}
\begin{proof}
  We will inherit some notation and terminology from \cite{haeupler2021universally} and adapt the proofs, as recreating them from scratch would yield an unwieldy paper.

  We first define that a demand $d = \{(s_i, t_i)\}_{i=1}^k$ is \textbf{connectable} if there exists a set of \emph{node-disjoint paths} $p_1, \ldots, p_k$ such that the endpoints of $p_i$ are $s_i$ and $t_i$.  

  Next, there exists a (black-box for our purposes) graph parameter $\SQ(G)$ called \textbf{shortcut quality}, defined in \cite{haeupler2021universally}, which assigns to each undirected graph a number $\Z_{\ge 1}$. \cite{haeupler2021universally} proves that for any \emph{connectable} demand we have $\OPT(d) \le \SQ(G) \poly(\log n)$. Moreover, it proves that there exists an input to the part-wise aggregation problem that requires at least $\SQ(G) / \poly(\log n)$ time steps. Hence, solving part-wise aggregation in $\SQ(G) \poly(\log n)$ will yield universally-optimal algorithms, which we achieve by reducing it to $\poly(\log n)$ many calls to point-to-point permutation routing on connectable demands, which is done via \Cref{thm:intro-routing}.

  However, exactly this is done by the proof of Lemma 7.2 in \cite{haeupler2021universally} (arXiv version) in the randomized case. We can directly remove the randomness by using Cole-Vishkin difference-label derandomization~\cite{cole1986deterministic}. %\alert{write out the full proof if there is any time for this!!}
\end{proof}

We are essentially done: the long line of work on low-congestion shortcuts have shown that, given an oracle that can solve the part-wise aggregation problem, one can solve many important global distributed problem with $\poly(\log n)$ calls to the oracle and an additional $D_G \cdot \poly(\log n)$ time steps of (supported-)CONGEST. Combining these together, we get one of our headline results.

\thmMainSupported*
\begin{proof}
  We will inherit some notation and terminology from \cite{ghaffari2021universally} and adapt the proofs, as recreating them from scratch would yield an unwieldy paper.  

  Theorem 17 of \cite{ghaffari2021universally} (arXiv version) exactly shows that, given a deterministic part-wise aggregation algorithm, we can get a deterministic supported-CONGEST simulation of the Minor-Aggregation model with $\poly(\log n)$ overheads. Furthermore, the paper shows how to give a deterministic MST in $\poly(\log n)$ time steps of the Minor-Aggregation model. Furthermore, the paper \cite{2022sssp} shows how to get a deterministic $\poly(\log n, 1/\eps)$-time-step Minor-Aggregation algorithm for the $(1+\eps)$-single-source shortest path algorithm. Combining these together, we obtain the result.
\end{proof}

Additionally, we mention that several other problems have been reduced to part-wise aggregation, but the reductions are not deterministic. These include exact min-cut~\cite{ghaffari2021universally} and Laplacian solving~\cite{anagnostides2021almost}. However, for min-cut, we know how to derandomize the main subproblem of finding the exact 2-respecting tree cut; the barrier lies in the tree packing procedure which hasn't been resolved even for much more powerful models.

\bibliographystyle{alpha}
\bibliography{refs} 

\appendix

\section{Routing Model Scheduling With Return}\label{sec:returnscheduling}

\antti{This section can use some more explanations \textbf{for the final version}, to better illustrate connection between the two models. Hash function should map from (path, identifier) pairs again.}

To avoid packet states in the routing algorithm, we need a scheduling algorithm that supports returning: all packets not delivered during the algorithm return to the start of their paths. %Support for returning requires not tolerating arbitrary starting positions, but non-delivered packets returning guarantees that packets start at the start of their paths, thus this is not an issue.

We modify the stateless weak scheduling algorithm (\Cref{alg:weakschedulingsubroutine}) to support returning undelivered packets. We double the length of the algorithm, and spend the last $2Ll$ time steps returning packets not delivered to the start of their paths. Since packets only move when the edge they would move over is not overcongested and when their virtual time is tight, this is simple.

\Cref{alg:weakreturnschedulingsubroutine} is a scheduling algorithm in the routing model. It assumes packets have paths $\pth(\pa)$ that are not properties of the packet itself, but computable quantities from the packet's properties: its source, sink, unique identifier and the current time step $t$ and the graph topology. Indeed, in the routing algorithm for \zeroonedemands, we assign each packet $\pa$ the path $P(\src(\pa), \tar(\pa))_j$, where $j$ is based on the current time step $t$.

\begin{algorithm}[H]
    \caption{Weak Scheduling Subroutine With Return}
    \label{alg:weakreturnschedulingsubroutine}
    \begin{algorithmic}[1]
    \State \textbf{Preconditions:} We are working in the local rule routing model, but additionally require that every packet has a computable, invariant path $\pth(\pa)$, and that these paths are simple.
    \State \textbf{Notation:} $\mathrm{pre}(\pa, v)$ is the edge leading to $v$ on $\pth(\pa)$ and $\mathrm{nxt}(\pa, v)$ is the edge from $v$ on the path. $\ind(p, v)$ is the index of node $v$ on a simple path $p$, and $\virt(p, v) := h(p) + \ind(p, v)$, where $h$ is the hash function from paths in $G$ to $\{0, \dots, L - 1\}$ received as input
%    \State \textbf{Global Inputs:} Large time step count $L$, small time step count $l$, hash function $h$ from paths in $G$ to $\{0, \dots, L - 1\}$, current time step $t$
    \Function{WeakReturnScheduler}{v, L, l, h, t}
    \If{$t < 2Ll$}

    \State Let $T = \lfloor t / l \rfloor$
    % \State Let $\mathrm{rem} = l - (t - Tl)$
    \For{All edges $e = (v, u)$ from $v$}
    \State Let $Q = \{\pa \in \que(v)_t : \mathrm{nxt}(\pa, v) = e \text{ AND } \virt(\pth(\pa), v) = T\}$
    \If{$0 < |Q| \leq l$}
    \State Forward an arbitrary $\pa \in Q$ over $e$
    \EndIf
    \EndFor

    \Else

    \State Let $T = \lfloor (4Ll - 1 - t) / l \rfloor$
    \For{All edges $e = (v, u)$ from $v$}
    \State Let $Q = \{\pa \in \que(v)_t : \mathrm{pre}(\pa, v) = \rev(e) \text{ AND } \virt(\pth(\pa), u) = T\}$
    \State If $Q$ is nonempty, forward an arbitrary $\pa \in Q$ over $e$
    \EndFor
    
    \EndIf
    \EndFunction
    \end{algorithmic}
\end{algorithm}

\begin{lemma}\label{lem:identicalbehaviourlemma}
    For a set of packets $\ps$ starting at the first edges of their paths $\pth(\pa)$ and fixed parameters $L$ and $l$, the set of packets delivered in $2Ll$ time steps by \Cref{alg:weakreturnschedulingsubroutine} (in the routing model) and \Cref{alg:weakschedulingsubroutine} (in the job-shop scheduling model applied to packet scheduling) is exactly the same, no matter how the choice on which eligible packet to forward over an edge is done.

    Furthermore, after $4Ll$ time steps, every packet not delivered by \Cref{alg:weakreturnschedulingsubroutine} is at its source $\src(\pa)$.
\end{lemma}

\shyr{Do we want to say something about the correspondence between packets and paths etc? We are using this implicitly} \antti{can explain more for full version, it's good for now} \antti{also we cheat because these hash functions don't take the identifier}

\begin{proof}
    The forwarding part of \Cref{alg:weakschedulingsubroutine} and \Cref{alg:weakreturnschedulingsubroutine} is exactly the same, and the decisions on which packets to forward do not matter, as within each large time step, every packet in $Q$ is forwarded over the edge (if initially $|Q| \leq l$) or none of the packets in $Q$ are.

    We prove the second claim by induction: the last $2Ll$ time steps rollback large time steps one by one, starting from the last one. The induction claim is as follows: at the start of a large rollback time step $T = (4Ll - 1 - t) / l$ (when $t \geq 2Ll$ and $t = 0\ (\text{mod } l)$), every packet is either at its destination or exactly where it was after the last time step of large time step $T$ in the forward direction.

    This holds for the last large time step, as it is rollbacked right after it occurs, and packets certainly do not move between small time steps. Finally, for the induction step, note that at large rollback time step $T$, the packets that want to return over an edge $e$ are exactly the packets that passed over $\rev(e)$ during large time step $T$, minus those already delivered. Thus, there are at most $l$ of them, and all of them can be returned over the edge within the rollback large time step.

    The algorithm can determine the next and previous edges on paths, as the paths are guaranteed to be simple.
\end{proof}

The noisy scheduling algorithm with return works exactly like the noisy scheduling algorithm in the scheduling model with \Cref{alg:weakreturnschedulingsubroutine} replacing \Cref{alg:weakschedulingsubroutine}. 

\begin{algorithm}[H]
    \caption{Local Rule Weak Noisy Scheduling Algorithm With Return}
    \label{alg:returnnoisyscheduler}
    \begin{algorithmic}[1]
    \State \textbf{Preconditions:} We are working in the local rule routing model, but additionally require that every packet has a computable, invariant path $\pth(\pa)$, and that these paths are simple
%    \State \textbf{Global Inputs:} Noise level $\beta$, domain path set $\calP$ satisfying $\pth(\pa) \in \calP$ for all $\pa$, optimal completion time scale $L$
%    \State \textbf{Abuse of Model:} We state the algorithm iteratively for convenience only, it is memoryless.
    \Function{ReturnNoisyScheduler}{v, $\beta$, $\calP$, $T$}
    \State Let $L$ be the minimum power of two at least $T$
    \State Let $l := \lceil 150c \ln n / \ln \ln n \rceil$
    \State Let $l' := 4\beta l$
    % \State Let $k := \lceil 16(\ln |\calP| + \ln n) / \ln l \rceil$
    \State Let $k := \lceil 8(b + 1)(2c + 1) \ln |M| / \ln l \rceil$ \Comment{$b$ controls success probability}
    \State Sample a set of $k$ independent, uniformly random hash functions $h_1, \dots, h_k$ (from $\calP$ to $\{0, \dots, L - 1\}$) using the shared randomness
    \For{$j \in \{1, \dots, k\}$}
    \For{$t \in \{0, \dots, 4Ll' - 1\}$}
    \State call $\mathrm{WeakReturnScheduler}(v, L, l', h_j, t)$
    \EndFor
    \EndFor
    \EndFunction
    \end{algorithmic}
\end{algorithm}

\returnnoisyschedulingworks

\begin{proof}
By \Cref{lem:iterategoodprob}, the probability that for any scale $L$ the set of $k$ hash functions is $(L,l)$-good for $\calP$ (slightly abusing the definition as the hash function is mapping from paths, not (path, identifier) pairs, but this only makes the resulting probability higher) is at least $1 - 2n^{-b}$ for an arbitrarily large constant $b$ of our choosing. As there are only logarithmically many scales $L$, the set of hash functions at every scale is good with high probability.

Now consider an arbitrary packet set $\ps$ and a $(\beta, T)$-good subset $\ps_S$. Since the set of hash functions for $L$ is $(L, l)$-good, there is a $(L, l)$-good hash function for $\ps_S$. By \Cref{lem:identicalbehaviourlemma} and \Cref{lem:scaleupstillgood}, at least a fourth of the packets in $\ps_S$ are delivered in the call to the weak return scheduler with that hash function. Packets not delivered return to their sources, as the algorithm only consists of some amount of calls to the weak scheduling subroutine with return, and that subroutine guarantees packets not delivered return to the first nodes of their paths.
\end{proof}

\section{Deferred Proofs}\label{sec:deferred}

\badpatterncount*

Recall first the definition of a bad pattern:

\badpatterndefinition*

\begin{proof} (of \Cref{lem:badpatterncount})
    We can select a bad pattern in two parts: first, we select the sizes $s_{T, m} = |B_{T, m}|$ of the sets, then, we select the contents. For the latter, we ignore the requirement that no job $\js$ can appear twice. As the number of pairs $(\js, i)$ satisfying $\seq(\js)_i = m$ for any machine $m$ is at most $C(\js)$, and the size of any nonempty $B_{T, m}$ is greater than $l$, there are then at most
    \begin{equation*}
        \prod_{T, m} \binom{C(\js)}{s_{T, m}} \leq \prod_{s_{T, m} \neq 0} \left(\frac{e C(\js)}{s_{T, m}}\right)^{s_{T, m}} \leq \prod_{T, m} \left(\frac{e C(\js)}{l}\right)^{s_{T, m}} = \left(\frac{e C(\js)}{l}\right)^{s}
    \end{equation*}
    ways to select the contents of the sets $B_{T, m}$. For the former, we first select the number $k$ of nonzero $s_{T, m}$ (at most $\lfloor s / (l + 1) \rfloor \leq \lfloor s / l \rfloor$ as any set is either nonempty or has size greater than $l$), then select for which $(T, m)$ $s_{T, m}$ is nonzero ($\binom{2L|M|}{k}$ ways), then divide the total size $s = \sum s_{T, m}$ among those $k$ (splitting $a$ identical items into $b$ groups can be done in $\binom{a + b - 1}{b - 1} \leq \binom{a + b}{b}$ ways). Thus, the sizes can be selected in at most
    \begin{equation*}
        \sum_{k = 1}^{\lfloor s / l \rfloor} \binom{2L|M|}{k} \binom{s + k}{k} \leq \sum_k (2L|M|(s + k))^k \leq 2 (2L|M|(s + \lfloor s / l \rfloor))^{\lfloor s / l \rfloor} \leq 2 (4L|M|s)^{s / l}
    \end{equation*}
    many ways. As $l \geq 150c \ln |M| / \ln \ln |M|$, $2|M|^c \geq L$ and $|M| \geq 32$, $c \geq 1$, we have
    \begin{equation*}
    \frac{4}{l} \ln(4L|E||\js|) \leq \frac{4}{l} \ln(8|M|^{3c}) = \frac{12c}{l} \ln |M| + \frac{12}{l} \leq \frac{1}{12} \ln \ln |M| + 1 \leq \ln |M|
    \end{equation*}
    thus, the number of ways to select the sizes is at most
    \begin{align*}
        2 (4L|M|s)^{s / l} &\leq 2 (4L|M||\js|)^{|\js| / (2l)} = 2 \exp(|\js| / (2l) \ln(4L|M||\js|))\\
        &\leq 2\exp\left(\frac{|\js|}{8} \cdot \frac{4}{l} \ln(4L|M||\js|)\right)\\
        &\leq 2\exp\left(\frac{|\js| \ln l}{8}\right)
    \end{align*}
    and, as desired, the number of bad patterns of size $s$ is at most
    \begin{equation*}
    2\exp\left(\frac{|\js| \ln l}{8}\right) \left(\frac{eC(\js)}{l}\right)^s.
    \end{equation*}
\end{proof}

\end{document}